\documentclass[11pt,manuscript,screen,review]{article}

\usepackage{indentfirst}  
\usepackage{algorithm, algorithmicx}
\usepackage{float}

\usepackage[utf8x]{inputenc} 
\usepackage{graphicx}
\usepackage{xcolor}

\usepackage{amsmath, amsthm, verbatim, enumerate, bbm, color}
\usepackage{amssymb} 

\usepackage{tikz}
\usetikzlibrary{calc}
\usepackage{tabularx}
\usetikzlibrary{arrows,automata}

\usepackage{hyperref}

\allowdisplaybreaks
\expandafter\let\expandafter\oldproof\csname\string\proof\endcsname
\let\oldendproof\endproof
\renewenvironment{proof}[1][\proofname]{%
	\oldproof[\bf #1]%
}{\oldendproof}

\allowdisplaybreaks

\theoremstyle{plain}
\newtheorem{theorem}{Theorem}
\newtheorem{lemma}{Lemma}[section]

\newtheorem{observation}[lemma]{Observation}
\newtheorem{corollary}[theorem]{Corollary}

\newtheorem{definition}[lemma]{Definition}

\newcommand{\inj}{\mathrm{inj}}

\newcommand{\aut}{\mathrm{aut}}
\RequirePackage[normalem]{ulem} 
\RequirePackage{color}\definecolor{RED}{rgb}{1,0,0}\definecolor{BLUE}{rgb}{0,0,1} 

\title{Counting Homomorphic Cycles in Degenerate Graphs}

\author{Lior Gishboliner\thanks{ETH Zurich. 
		Email: lior.gishboliner@math.ethz.ch.}
	\and Yevgeny Levanzov\thanks{School of Mathematics, Tel Aviv University, Tel Aviv 69978, Israel. Email: yevgenyl@mail.tau.ac.il. Supported in part by ERC Consolidator Grant 863438 and NSF-BSF Grant 20196.}
	\and Asaf Shapira\thanks{School of Mathematics, Tel Aviv University, Tel Aviv 69978, Israel. Email: asafico@tau.ac.il. Supported in part by ISF Grant 1028/16, ERC Consolidator Grant 863438 and NSF-BSF Grant 20196.}
	\and Raphael Yuster\thanks{Department of Mathematics, University of Haifa, Haifa 3498838, Israel. Email: raphael.yuster@gmail.com. Supported in part by ISF Grant 1028/16.}}

\begin{document}
	
	\maketitle
	
	\begin{abstract}
		
		Since counting subgraphs in general graphs is, by and large, a computationally demanding problem, it is natural to try and design fast algorithms for restricted families of graphs. One such family that has been extensively studied is that of graphs of bounded degeneracy (e.g., planar graphs). This line of work, which started in the early 80's, culminated in a recent work of Gishboliner et al., which highlighted the importance of the task of counting homomorphic copies of cycles (i.e., cyclic walks) in graphs of bounded degeneracy.
		
		Our main result in this paper is a surprisingly tight relation between the above task and the well-studied problem of {\em detecting (standard) copies} of directed cycles in {\em general directed} graphs. More precisely, we prove the following:
		
		\begin{itemize}
			\item One can compute the number of homomorphic copies of $C_{2k}$ and $C_{2k+1}$ in $n$-vertex graphs of bounded degeneracy in time $\tilde{O}(n^{d_{k}})$, where
			the fastest {\em known} algorithm for detecting directed copies of $C_k$
			in general $m$-edge digraphs runs in time $\tilde{O}(m^{d_{k}})$.
			
			\item Conversely, one can transform any $O(n^{b_{k}})$ algorithm for computing the number of homomorphic copies of $C_{2k}$ or of $C_{2k+1}$ in $n$-vertex graphs of bounded degeneracy, into an $\tilde{O}(m^{b_{k}})$ time algorithm for detecting directed copies of $C_k$ in general $m$-edge digraphs.
			
		\end{itemize}
		
		We emphasize that our first result does not use a black-box reduction (as opposed to the second result which does). Instead, we design an algorithm for computing the number of $C_k$-homomorphisms in degenerate graphs and show that one part of its {\em analysis} can be reduced to the analysis of the fastest known algorithm for detecting directed cycles in general digraphs, which was carried out in a recent breakthrough of Dalirrooyfard, Vuong and Vassilevska Williams. 
		As a by-product of our algorithm, we obtain a new
		algorithm for detecting $k$-cycles in directed and undirected graphs of bounded degeneracy that is faster than
		all previously known algorithms for $7 \le k \le 11$, and faster for all $k \ge 7$ if the matrix multiplication exponent is $2$.
	\end{abstract}


	\section{Introduction}
	Counting occurrences of small subgraphs in a given input graph is among the most fundamental algorithmic problems. 
	Most prominently, it is the subject of a rich line of research in parameterized complexity theory 
	\cite{CDM,C_Marx,DJ,FG,Grohe,Marx,Meeks,NP,V W,VW-W}, which has by now produced several important general results on the fixed-parameter tractability of subgraph counting problems \cite{DJ,Grohe,Marx}. Works in this area are too numerous to survey here, so we refer the reader to the above-cited papers for further references.
	On the practical side of things, subgraph counting is important due to the role played by subgraph counts in the analysis of real-world networks, see, e.g., \cite{MSIKCA,Przulj}, and the references therein.
	
	\subsection{Detecting and Counting Cycles}
	One of the most fundamental and well-studied cases of subgraph-counting is that of counting (and detecting) cycles. Such questions have been studied in both undirected and directed graphs. We will always denote by $n$ the number of vertices of the input graph and by $m$ its number of edges. 
	The problem of counting copies of a $k$-length cycle is known to be $\#W[1]$-hard \cite{FG},
	meaning that it is unlikely to admit an algorithm which runs in time 
	$f(k) \cdot n^{O(1)}$ for any computable function $f$. 
	For cycle detection, however, efficient algorithms exist and are the subject of numerous works \cite{AYZ_ColorCoding,AYZ,EG,IR,L-VW-W,Monien,YZ,YZ_2}.  
	Notably, Alon, Yuster and Zwick \cite{AYZ} presented several cycle detection algorithms for both the undirected and directed settings, improving upon earlier results of Itai and Rodeh \cite{IR} and Monien \cite{Monien}. Some of these algorithms used the technique of {\em color-coding}, introduced earlier in \cite{AYZ_ColorCoding}. Most relevant for us is the problem of detecting directed cycles in digraphs,
	where \cite{AYZ} provided a classical algorithm that detects a directed $k$-cycle (for a fixed $k$) in an $m$-edge directed graph running \nolinebreak in \nolinebreak time 
	\begin{equation}\label{eq:combinatorial_cycle_detect_runtime}
	O(m^{2-1/\lceil k/2 \rceil}).
	\end{equation}
	This is still the fastest {\em combinatorial} algorithm
	for directed $k$-cycle detection in terms of $m$; see the remark at the end of Section \ref{subsec:related_work} concerning its optimality.
	%
	However, faster algorithms exist if we use algebraic algorithms relying on fast matrix multiplication.
	Let $\omega$ denote the exponent of fast matrix multiplication, namely the infimum over all constants $t$ such that two $n \times n$ matrices can be multiplied using $\tilde{O}(n^t)$ field operations. It is known that
	$\omega < 2.373$ \cite{AV,Legall,V W 2}.
	Indeed, motivated by a simple $O(m^{2\omega/(\omega+1)})$ time algorithm of \cite{AYZ} for detecting a $3$-cycle, Yuster and Zwick \cite{YZ} set out to use fast matrix multiplication for $k$-cycle detection (for a fixed $k$).
	They obtained an algorithm running in time $\tilde{O}(m^{c_k})$, and were able to show that for $k=4,5$,
	the value of $c_k$ (which they explicitly computed for $k=4,5$) is strictly smaller than $2-1/\lceil k/2 \rceil$, thereby improving upon the aforementioned combinatorial algorithm. It is, to date, the fastest algorithm for directed $k$-cycle detection when $k=4,5$ in terms of $m$.
	For $k \ge 6$, the exact value of $c_k$ turned out to be hard to analyze, due to an involved use of dynamic programming, which resulted in a complicated recursive relation. This led Yuster and Zwick to raise a conjecture\footnote{In fact, they only conjectured the value of $c_k$ when $k$ is odd; for even $k$ it was not even clear what the ``right'' conjecture should be.} regarding the runtime of their algorithm (i.e., the value of $c_k$). 
	This conjecture was recently resolved by Dalirrooyfard, Vuong and Vassilevska Williams \cite{DVW}, establishing that the algorithm of Yuster and Zwick detects a directed $k$-cycle in a graph with $m$ edges in time
	$\tilde{O}(m^{c_k})$, \nolinebreak where\footnote{As in \cite{DVW}, the value of $c_k$ in \eqref{eq:c_k-bounds} assumes, for simplicity, that rectangular matrix multiplication is simulated by square matrix multiplication; the exact definition of $c_k$ is given in \eqref{eq:small-c_k}. If $\omega > 2$ and the fastest known rectangular matrix multiplication algorithms are used, $c_k$ results in a marginal improvement over the value in \eqref{eq:c_k-bounds}.}
	\begin{align}\label{eq:c_k-bounds}
	c_k & \le \frac{(k+1)\omega}{2\omega + k-1} \qquad\qquad  \textup{if } k \textup{ is odd,} \\
	c_k & \le \frac{k\omega - \frac{4}{k}}{2\omega + k-2-\frac{4}{k}} \qquad \, \textup{if } k \textup{ is even.} \nonumber
	\end{align}
	It should be noted that equality holds in the odd case when $\omega \le \frac{2k}{k-1}$\footnote{When $k$ is odd and $\omega > \frac{2k}{k-1}$, one can prove that $c_k \le 2-\frac{2}{k+1} < \frac{(k+1)\omega}{2\omega + k-1}$, which matches the runtime of the fastest combinatorial algorithm.}.
	Moreover, it was proved in \cite{DVW} that equality holds also in the even case if $\omega=2$
	and the exact value of $c_6$ was determined for all $\omega$.
	To date, for any pair $(k,\omega)$ with $k \ge 3$ and $2 \le \omega < 2.373$, the fastest algorithm
	for detecting a directed $k$-cycle in a directed $m$-edge graph (i.e., in terms of $m$ alone) is either the combinatorial
	$O(m^{2-1/\lceil k/2 \rceil})$ algorithm or the $\tilde{O}(m^{c_k})$ algorithm (which of them is faster depends on $(k,\omega)$). 
	It is easy to verify that if $\omega > 2$, then for all large enough $k$ the 
	$O(m^{2-1/\lceil k/2 \rceil})$ is faster, while if $\omega = 2$ then $\tilde{O}(m^{c_k})$ is always faster, and $c_k$ is approximately $2 - \frac{4}{k}$ in this case (for $k$ large). 
	
	\subsection{Counting Homomorphisms in Degenerate Graphs}
	In this work we are concerned with the problem of counting {\em homomorphisms} of cycles (i.e., cyclic walks of a given length) in graphs of {\em bounded degeneracy}\footnote{We note that counting {\em copies} of cycles remains $\#W[1]$-hard even in bounded-degeneracy graphs. This follows from the fact that counting $k$-cycles in general graphs can be reduced to counting $2k$-cycles in $2$-degenerate ones via the reduction which replaces each edge of the input graph by a path of length $2$ (thus producing a $2$-degenerate graph).}. 
	Arguably, homomorphism-counting is the most basic subgraph counting problem, to which other natural problems --- such as counting copies or induced copies --- can be reduced \cite{CDM}.  
	
	Recall that a graph is called {\em $d$-degenerate} if it admits a vertex ordering $v_1,\dots,v_n$ such that $v_i$ has at most $d$ neighbors among $v_{i+1},\dots,v_n$ (for every $1 \leq i \leq n$). 
	In the setting considered here, we look for algorithms which run in time $f(d) \cdot n^{\alpha}$, where $d$ is the degeneracy\footnote{We note that degeneracy is closely related to another well-studied graph parameter, namely arboricity, which is the minimum number of forests into which the edge-set of a graph can be partitioned. It is well-known that the arboricity of a $d$-degenerate graph is between $(d+1)/2$ and $d$.} of the input graph, $n$ is its number of vertices, and $\alpha$ does not depend on $d$. Such algorithms are particularly useful when the input graph has bounded degeneracy. The family of such graphs is rich, including for example all non-trivial minor-closed graph classes (in particular, planar graphs), 
	preferential attachment graphs \cite{BA}, and bounded expansion graphs \cite{NDeM}.  
	
	The study of subgraph counting in degenerate graphs goes back to a classical work of Chiba and Nishizeki \cite{CN} from the early 1980s. Very recently, this research direction has seen several substantial developments \cite{Bera-Sesh-GLS,BPS,BPS_2,Bressan}.  
	Bressan \cite{Bressan} gave a general algorithm for counting homomorphisms in bounded-degeneracy graphs, which in particular implies a sufficient condition\footnote{Bressan's result in its full generality states that for every graph $H$, one can count $H$-homomorphisms in bounded-degeneracy graphs in time $\tilde{O}(n^{\tau(H)})$, where $\tau(H)$ is a suitable width parameter called {\em DAG treewidth}.} (on graphs $H$) under which $H$-homomorphisms can be counted in time $\tilde{O}(n)$. 
	In \cite{Bera-Sesh-GLS}, it was shown that this sufficient condition is also necessary, and moreover, a clean combinatorial characterization of the graphs $H$ satisfying it was obtained. Specifically, it was shown that $H$-homomorphisms in bounded-degeneracy graphs can be counted in time $\tilde{O}(n)$ if and only if $H$ contains no induced cycles of length $6$ or larger\footnote{The ``only if" direction of this result is under a certain standard hardness assumption from fine-grained complexity \cite{A-VW}. The proof of this direction actually shows that if $H$ contains an induced $\ell$-cycle for some $\ell \geq 6$, then one cannot count $H$-homomorphisms in bounded-degeneracy graphs in time $O(n^{1+\gamma})$, where $\gamma > 0$ is some fixed constant.}. A similar result has been independently obtained in \cite{BPS_2}. These results highlight the important role played by cycles of length at least $6$, as being the minimal graphs whose homomorphisms cannot be counted in (almost) linear time (in bounded-degeneracy graphs).
	This motivates the study of the time complexity of counting homomorphisms of such cycles (with the intention of improving the running time of Bressan's general algorithm for cycles), which is precisely the question studied in the present paper. 
	
	\subsection{Our Results}
	Let $C_k$ denote the cycle of length $k$ (where $k$ is fixed throughout the paper).
	Here we state our results regarding the problem of counting $C_k$-homomorphisms in graphs of bounded degeneracy. 
	By the aforementioned result of Bressan \cite{Bressan}, if $3 \leq k \leq 5$ then $C_k$-homomorphisms can be counted in near-linear time (which is optimal), and we therefore focus on the case $k \geq 6$. 	
	As our main results in this paper show, the problem of counting $C_{2k}$- and $C_{2k+1}$-homomorphisms in bounded-degeneracy graphs is related to the problem of detecting directed $k$-cycles in general (i.e., not necessarily degenerate) digraphs. In particular, we shall design two algorithms:
	a combinatorial algorithm for counting cycle homomorphisms whose runtime corresponds to the runtime of the combinatorial $k$-cycle detection algorithm of \cite{AYZ} and 
	an algorithm for counting cycle homomorphisms whose runtime obeys the same recursive relation as the aforementioned directed-cycle detection algorithm of Yuster and Zwick \cite{YZ}. Combining these two algorithms and the work of Dalirrooyfard, Vuong and Vassilevska Williams \cite{DVW}, we will obtain the following theorem, which is our main result in this paper. Let $\textsc{hom-cnt}_{H}$ denote the problem of computing the number of $H$-homomorphisms. 
	
	\begin{theorem}[Main Result]\label{thm:main}
		For every $k \geq 2$, the following hold:\\
		(i) $\textsc{hom-cnt}_{C_{2k}}$ and $\textsc{hom-cnt}_{C_{2k+1}}$ can be solved in bounded-degeneracy graphs in time $\tilde{O}(n^{c_k})$, where $c_k$ is as defined in \eqref{eq:small-c_k}.\\
		(ii) $\textsc{hom-cnt}_{C_{2k}}$ and $\textsc{hom-cnt}_{C_{2k+1}}$ can be solved in bounded-degeneracy graphs in time $\tilde{O}(n^{2-1/\lceil k/2 \rceil})$.
	\end{theorem}
	The algorithm of case (i) in Theorem \ref{thm:main} combines the general structure of the Yuster-Zwick algorithm with several significant additional twists used in the setting of degenerate graphs, inspired by the approach in \cite{Bressan}. 
	
	As opposed to the algorithm of case (i) which relies on fast matrix multiplication, the algorithm of case (ii) is purely combinatorial.
	Furthermore, it can be modified to give a combinatorial algorithm for counting $C_k$-homomorphisms in general (i.e., not necessarily degenerate) graphs, both directed and undirected, which runs in time $\tilde{O}(m^{2-1/\lceil k/2 \rceil})$. This algorithm can in turn be used to obtain a new combinatorial algorithm for directed $k$-cycle detection, which essentially matches the runtime of the best such combinatorial algorithm to date, namely the algorithm of \cite{AYZ} mentioned above, while being somewhat simpler. It should be noted that another combinatorial cycle-detection algorithm with the same runtime was given in \cite{VW-W-Y} (see also \cite{L-VW-W}). 
	The details appear in Section \ref{sec:general_graphs}.   
	
	
	The relation to directed-cycle detection is in fact more robust, as is evidenced by our next result, which states that counting $C_{2k}$- or $C_{2k+1}$-homomorphisms already in $2$-degenerate graphs is {\em at least as hard} as detecting a directed $k$-cycle in general digraphs. Note that if $k=2$ then both problems can be solved in time linear in the number of edges, so the interesting case is when $k \ge 3$. Formally, we prove the following theorem.
	
	\begin{theorem}\label{thm:count-vs-detect}
		Let $k \ge 3$. If there is an algorithm that computes $\textsc{hom-cnt}_{C_{2k}}$ or $\textsc{hom-cnt}_{C_{2k+1}}$ in $2$-degenerate graphs in $O(n^\alpha)$ time, then there is an algorithm that decides if an arbitrary directed $m$-edge graph has a directed $k$-cycle in time $O(m^\alpha)$ (randomized) and $O(m^\alpha \log m)$ (deterministic).
	\end{theorem}
	
	To make the relation between  directed-cycle detection in general graphs and counting cycle homomorphisms in bounded-degeneracy graphs even sharper, recall that for all $k \ge 3$, the presently fastest $k$-cycle detection algorithm in $m$-edge directed graphs is either the $O(m^{2-1/\lceil k/2 \rceil})$ algorithm \cite{AYZ} or the $\tilde{O}(m^{c_k})$ algorithm \cite{DVW,YZ}. Let us therefore define for a given $2 \le \omega < 2.373$
	$$
	d_k = \min\{c_k \,,\,2-1/{\lceil k/2 \rceil}\}\;.
	$$
	Recall that if $\omega > 2$ then for all large enough $k$ we have $d_k = 2-1/{\lceil k/2 \rceil}$ (i.e., the second term is smaller), while if $\omega = 2$ then $d_k = c_k$ for all $k \geq 3$.
	The following corollary immediately follows from Theorem \ref{thm:main} and Theorem \ref{thm:count-vs-detect}.
	\begin{corollary}\label{cor:of-thm-count-vs-detect}
		Let $k \ge 3$.
		If the fastest algorithm (in terms of the number of edges $m$) for deciding if a directed graph has a directed $k$-cycle runs in $\tilde{O}(m^{d_k})$ time then the fastest algorithm for computing  $\textsc{hom-cnt}_{C_{2k}}$ in degenerate graphs runs in $\tilde{O}(n^{d_k})$ time and the fastest algorithm for computing  $\textsc{hom-cnt}_{C_{2k+1}}$ in degenerate graphs runs in $\tilde{O}(n^{d_k})$ time.
	\end{corollary}

	In particular, for every $\varepsilon > 0$, if we 
	can improve the bound in Theorem \ref{thm:main} and compute, say, $\textsc{hom-cnt}_{C_{2k}}$
	in time $O(n^{d_k-\varepsilon})$, then we could decide if a digraph with $m$ edges has a directed $k$-cycle in time $O(m^{d_k-\varepsilon})$, which is faster (in terms of $m$) than the presently fastest known algorithm for the latter problem.
	
	Another interesting by-product of Theorem \ref{thm:main} is that it can be used (together with a standard application of the color-coding method) to give an algorithm for (directed or undirected) $k$-cycle detection
	in bounded-degeneracy graphs.
	\begin{theorem}\label{thm:detect-degeneracy}
		Let $k \ge 6$. There is an algorithm that detects if a (directed) bounded-degeneracy graph has a (directed) $k$-cycle in $\tilde{O}(n^{d_{\lfloor k/2 \rfloor}})$ time. 
	\end{theorem}
	As can easily be verified by examining the values of $c_k$ (hence $d_k$), the algorithm of Theorem
	\ref{thm:detect-degeneracy} is faster than any known algorithm for cycle detection in degenerate graphs for $7 \le k \le 11$.
	If $\omega = 2$, it is faster for all $k \ge 7$. In any case it is never slower (up to polylogarithmic factors) than the previously fastest algorithm for this problem, given in \cite{AYZ}, Theorem 4.2.
	
	Let us give some more details of the proof of Theorem \ref{thm:main}. For (di)graphs $G,H$, let us denote by $\hom(H,G)$ the number of homomorphisms from $H$ to $G$. Suppose that we want to compute $\hom(C_{\ell},G)$ for a given $d$-degenerate input graph $G$. As is usual when working with degenerate graphs, we approach this task by first finding a degeneracy ordering\footnote{A degeneracy ordering of a graph can be found in time linear in its number of edges \cite{MB}. In particular, since a bounded-degeneracy graph has $O(n)$ edges, a degeneracy ordering of
		such a graph can be found in time $O(n)$.} $v_1,\dots,v_n$ of $G$ --- i.e., a vertex ordering in which each vertex $v_i$ has at most $d$ neighbors succeeding it --- and then considering the orientation $\vec{G}$ of $G$ in which all edges are oriented forward with respect to this ordering, namely, from $v_i$ to $v_j$ for $\{v_i,v_j\} \in E(G)$ with $i < j$. Observe that $\vec{G}$ can be obtained in linear time (in $|E(G)|=O(dn)$), that it is acyclic, and that all of its vertices have out-degree at most $d$. 
	Moreover, it is not hard to see that 
	$\hom(C_\ell,G) = \sum_{\vec{H}}{\hom(\vec{H},\vec{G})}$, where $\vec{H}$ runs over all acyclic orientations of $C_{\ell}$. It follows that counting $C_{\ell}$-homomorphisms in $d$-degenerate graphs reduces to counting homomorphisms of (all) acyclic orientations of $C_{\ell}$ in directed
	acyclic graphs (DAGs) of maximum out-degree $d$. 
	
	As it turns out, there is a particular orientation of $C_{\ell}$ which is harder, in a sense, than all others; this is the orientation with the maximum number of sources, namely, the orientation where edges are oriented in an alternating fashion (with the exception of one edge in the case that $\ell$ is odd). We call this the {\em alternating orientation} of $C_{\ell}$. To establish that this is indeed the ``hardest orientation", we show (roughly speaking) that for every DAG $\vec{H}$ and for every {\em directed subdivision} $\vec{H}'$ of $\vec{H}$, counting $\vec{H}'$-homomorphisms can be reduced to counting $\vec{H}$-homomorphisms (in input DAGs of bounded maximum degree). It is easy to see that every acyclic orientation $\vec{H}$ of $C_{\ell}$ (with at least two sources) is a directed subdivision of the alternating orientation of $C_{\ell'}$ for some even $\ell' \leq \ell$, and that $\ell' = \ell$ if and only if $\vec{H}$ itself is alternating\footnote{Indeed, one can observe that if the number of sources in $\vec{H}$ is $p > 1$, then $\vec{H}$ is a directed subdivision of the alternating orientation of $C_{2p}$. 
		If $p = 1$ then we can think of $C_{\ell}$ as the subdivision of the 2-vertex alternating cycle $C_2$, which is the multigraph with two vertices $x,y$ and two parallel edges from $x$ to $y$. Theorem \ref{thm:reduction-subdivision_main} applies to $\vec{H} = C_2$ as well. 
		A different way to handle the case $p = 1$ is to observe that in this case
		one can easily count homomorphisms of $\vec{H}$ in linear time \cite{Bressan}.
	}. 
	It follows that the running time of counting $C_{\ell}$-homomorphisms is dominated by the running time of counting homomorphisms of alternating orientations of cycles of even length at most $\ell$. 
	
	The aforementioned general reduction for directed subdivisions is stated in the following theorem. For technical reasons, the reduction reduces the problem of counting $\vec{H'}$-homomorphisms to the problem of counting $\vec{H}$-homomorphisms in {\em weighted digraphs}. We denote by $\textsc{w-hom-cnt}_{\vec{H}}$ the weighted analogue of $\textsc{hom-cnt}_{\vec{H}}$; see Section \ref{sec:hom-and-subdiv} for the precise definition. 
	
	
	
	\begin{theorem}\label{thm:reduction-subdivision_main}
		Let $\vec{H}$ be a DAG and let $\vec{H'}$ be a directed subdivision of $\vec{H}$. If $\textsc{w-hom-cnt}_{\vec{H}}$ can be solved in time $O(n^\alpha)$ in $n$-vertex weighted degenerate digraphs\footnote{Here, and in what follows, by degenerate digraphs we mean digraphs where all out-degrees are bounded.}, then $\textsc{hom-cnt}_{\vec{H}'}$ can be solved in time $O(n^\alpha)$ in $n$-vertex degenerate digraphs.
	\end{theorem}
	We believe Theorem \ref{thm:reduction-subdivision_main} to be of independent interest. 
	With this theorem at hand, it remains to give an efficient algorithm which counts homomorphisms of the alternating orientation of $C_{2k}$ (in DAGs of bounded maximum out-degree) for every $k \geq 3$, which is the main step towards proving Theorem \ref{thm:main}. 
	Since applying Theorem \ref{thm:reduction-subdivision_main} requires that the algorithm works in the more general weighted setting, we need to modify it accordingly (this modification is straightforward and does not pose additional difficulties).
	
	\subsection{Related Work}\label{subsec:related_work}
	
	While the fastest algorithm for $k$-cycle detection in directed $m$-edge graphs runs in $\tilde{O}(m^{d_k})$ time, it is worth noting that there are other algorithms expressed in terms of both $n$ and $m$ (or $n$ alone).
	Indeed, the color-coding method \cite{AYZ_ColorCoding} shows that a simple 
	directed or undirected cycle of size $k$ in a directed or undirected graph can be detected in
	either $\tilde{O}(nm)$ time or $\tilde{O}(n^\omega)$ time. For dense graphs, this is faster than the
	$\tilde{O}(m^{d_k})$ time algorithm.
	Eisenbrand and Grandoni \cite{EG} proved that a directed $C_4$ can be detected in time $O(m^{2-2/\omega}n^{1/\omega})$.
	This algorithm is inferior to the $\tilde{O}(m^{c_k})$ algorithm for sparse graphs
	and inferior to the $O(n^\omega)$ algorithm for dense graphs, but is better than both in some intermediate
	\nolinebreak range.
	
	As for hardness,  Lincoln et al. \cite{L-VW-W} proved conditional lower bounds for $k$-cycle detection.
	Under a widely-believed assumption, $K_k$ cannot be detected faster than $O(C(n,k))$ time, where
	$C(n,k) = M(n^{\lceil k/3 \rceil},n^{\lfloor k/3 \rfloor},n^{\lceil (k-1)/3 \rceil})$ and
	where $M(a,b,c)$ is the fastest known runtime for multiplying an $a \times b$  by a $b \times c$ matrix.
	Assuming this, they proved that detecting a directed $k$-cycle in an $m$-edge graph requires
	$m^{\frac{2\omega k}{3(k+1)}-o(1)}$ time. The same reduction also shows that any {\em combinatorial} algorithm for detecting a directed $k$-cycle in an $m$-edge graph requires $m^{2-1/\lceil k/2 \rceil - o(1)}$ time, in this case under the suitable 
	hardness hypothesis that any combinatorial algorithm for $K_k$-detection requires $n^{k-o(1)}$ time. This lower bound of $m^{2-1/\lceil k/2 \rceil - o(1)}$ coincides with the runtime in \eqref{eq:combinatorial_cycle_detect_runtime}, showing that the algorithms of \cite{AYZ,L-VW-W,VW-W-Y}, as well as our combinatorial algorithm for directed-cycle detection given in Section \ref{sec:general_graphs}, are optimal (for combinatorial algorithms). The same applies to the algorithm given by Part (ii) \nolinebreak of \nolinebreak Theorem \nolinebreak \ref{thm:main}.
	
	\paragraph{Paper Overview:}
	The goal of Section \ref{sec:prelim} is to introduce the definitions we will use in subsequent sections. Section \ref{sec:hom-and-subdiv} is devoted to proving Theorem \ref{thm:reduction-subdivision_main}. 
	Theorem \ref{thm:main} is proved in Sections \ref{sec:comb} and \ref{sec:mat-mult}: the former deals with Part (ii) of the theorem and the latter with Part (i). Section \ref{sec:count-vs-detect} contains the proofs of Theorems \ref{thm:count-vs-detect} and \ref{thm:detect-degeneracy}. 
	Finally, in Section \ref{sec:general_graphs} we apply our methods to obtain combinatorial algorithms for counting $C_k$-homomorphisms and for $C_k$-detection in general graphs (both directed and undirected).  
	
	\section{Preliminaries}\label{sec:prelim}
	Here we introduce the basic definitions to be used throughout the paper. Recall that for (undirected) graphs $G,H$, a homomorphism from $H$ to $G$ is a map $\varphi : V(H) \rightarrow V(G)$ such that $\{\varphi(u),\varphi(v)\} \in E(G)$ whenever $\{u,v\} \in E(H)$. Similarly, for digraphs $\vec{G},\vec{H}$, a homomorphism from $\vec{H}$ to $\vec{G}$ is a map $\varphi : V(\vec{H}) \rightarrow V(\vec{G})$ such that $(\varphi(u),\varphi(v)) \in E(\vec{G})$ whenever $(u,v) \in E(\vec{H})$.
	We will consider (di)graphs with edge weights. Let us define what we mean by the {\em weight} of a homomorphism. 
	\begin{definition}\label{def:weight-hom}
		Let $\vec{H}$ be a digraph, let $\vec{G}$ be a weighted digraph with a weight function \linebreak $w_{\vec{G}}: E(\vec{G}) \rightarrow \mathbb{R}_{\geq 0}$, and let $\varphi: V(\vec{H}) \rightarrow V(\vec{G})$ be a homomorphism. The {\em weight} of $\varphi$, denoted $W(\varphi)$, is defined as follows:
		$$W(\varphi) := \prod_{(u,v) \in E(\vec{H})} w_{\vec{G}}((\varphi(u), \varphi(v))). $$
	\end{definition}
	
	For digraphs $\vec{G},\vec{H}$, denote by $\text{Hom}(\vec{H}, \vec{G})$ the set of all homomorphisms from $\vec{H}$ to $\vec{G}$. For a weight-function $w_{\vec{G}}$ on the edges of $\vec{G}$, define the {\em weighted homomorphism count} as follows:
	$$ \hom(\vec{H}, \vec{G}, w_{\vec{G}}) := \sum_{\varphi \in \text{Hom}(\vec{H}, \vec{G})} W(\varphi).$$
	As mentioned above, we denote by $\textsc{w-hom-cnt}_{\vec{H}}$ the problem of computing $\hom(\vec{H}, \vec{G}, w_{\vec{G}})$ for a given weighted input digraph $\vec{G}$. 
	(Usual --- i.e., unweighted --- homomorphism counts can be cast in this setting by letting all edge-weights be $1$.)
	
	\noindent
	Finally, we recall the notion of subdivision for directed graphs: 
	
	
	\begin{definition}\label{def:directed-graph-subdiv}
		Let $\vec{H}$ be a directed graph and write $E(\vec{H}) = \{e_1, \dots, e_t\}$. Let $(x_1, \dots, x_t)$ be a sequence of positive integers. The {\em directed $(x_1,\dots,x_t)$-subdivision} of $\vec{H}$ is the digraph obtained from $\vec{H}$ by replacing each edge $e_i$ with a directed path with $x_i$ edges, where paths replacing different edges are internally disjoint. We call $(x_1,\dots,x_t)$ the {\em subdivision sequence}. We say that a digraph $\vec{H}'$ is a {\em directed subdivision} of $\vec{H}$ if it is the directed $(x_1,\dots,x_t)$-subdivision of $\vec{H}$ for some $(x_1,\dots,x_t)$.  
	\end{definition}
	
	\section{Homomorphism Counting and Graph Subdivisions}\label{sec:hom-and-subdiv}
	
	In this section we prove Theorem \ref{thm:reduction-subdivision_main}, showing that counting (weighted) homomorphisms of a directed graph is at least as hard as counting homomorphisms of its directed subdivisions. The main tool in the proof is (a digraph variant of) an extremely useful result of Curticapean, Dell and Marx \cite{CDM}, stated below as Lemma \ref{lem:hom-linear-combination}. This lemma deals with computing {\em linear combinations of homomorphism counts}.  
	For digraphs 
	$\vec{H_1}, \dots, \vec{H_k}$ and non-zero constants $c_1, \dots ,c_k$, let 
	$\textsc{hom-cnt}_{c_1 \vec{H_1} + \dots + c_k \vec{H_k}}$ be the problem of computing 
	$c_1 \cdot \hom(\vec{H_1},\vec{G}) + \dots + c_k \cdot \hom(\vec{H_k},\vec{G})$ for an input digraph $\vec{G}$. Lemma \ref{lem:hom-linear-combination} states that solving 
	$\textsc{hom-cnt}_{c_1 \vec{H_1} + \dots + c_k \vec{H_k}}$ is essentially equivalent to solving $\textsc{hom-cnt}_{\vec{H_i}}$ for all $1 \leq i \leq k$. 

	\begin{lemma}\label{lem:hom-linear-combination}
		Let $\vec{H_1}, \dots, \vec{H_k}$ be pairwise non-isomorphic digraphs and let $c_1, \dots, c_k$ be non-zero constants. Then $\textsc{hom-cnt}_{c_1 \vec{H_1} + \dots + c_k \vec{H_k}}$ can be solved in time $O(n^\alpha)$ in $n$-vertex degenerate digraphs if and only if
		$\textsc{hom-cnt}_{\vec{H_i}}$ can be solved in time $O(n^\alpha)$ in $n$-vertex degenerate digraphs for all $1 \leq i \leq k$. 
	\end{lemma}
	
	The undirected version of Lemma \ref{lem:hom-linear-combination} was proved in \cite{CDM}, and subsequently used in \cite{Bera-Sesh-GLS} to study homomorphism-counting in degenerate graphs. The proof uses tensor products and a result of Erd\H{o}s, Lov\'{a}sz and Spencer \cite{ELS} regarding linear independence of homomorphism counts. Since the proof of the directed variant (namely, Lemma \ref{lem:hom-linear-combination}) is essentially the same as that of the undirected one, we postpone it to the appendix. 
	
	
	We are now in a position to prove Theorem \ref{thm:reduction-subdivision_main}.
	
	\begin{proof}[Proof of Theorem \ref{thm:reduction-subdivision_main}]
		Let $\vec{H}$ be a fixed digraph and write $E(\vec{H}) = \{e_1, \dots, e_t\}$.
		We assume that $\textsc{w-hom-cnt}_{\vec{H}}$ can be solved in time $O(n^\alpha)$ (for some $\alpha \geq 1$) in $n$-vertex weighted degenerate digraphs.	
		For an integer $p \geq 1$, we denote by $SD_{\vec{H}, p}$ the set of all directed subdivisions of $\vec{H}$ where each edge is replaced by a directed path of length at most $p$, and the original vertices (i.e., the vertices of $\vec{H}$) are labeled as in $\vec{H}$. We will prove the theorem simultaneously for all digraphs $\vec{H}' \in SD_{\vec{H}, p}$ (this is clearly sufficient as $p$ was arbitrary).
		Note that every digraph in $SD_{\vec{H}, p}$ is defined by a unique subdivision sequence $x=(x_1, \dots, x_t) \in [p]^t$. Our goal is to show that $\textsc{hom-cnt}_{\vec{H}'}$ can be solved in time $O(n^\alpha)$ in bounded-degeneracy digraphs for every digraph $\vec{H}' \in SD_{\vec{H}, p}$. To this end, we will show that a certain linear combination of $\hom(\vec{H}', \vec{G})$, where $\vec{H}'$ runs over all digraphs $\vec{H}' \in SD_{\vec{H}, p}$, can be computed in time $O(n^\alpha)$. We will then apply Lemma \ref{lem:hom-linear-combination} to complete the proof.
		
		Let $\vec{G}$ be an $n$-vertex digraph where all out-degrees are at most $\kappa = O(1)$. We construct a weighted digraph $\vec{F}$, that will have the same vertex set as $\vec{G}$, as follows. For every pair $x,y \in V(\vec{G})$ and every $\ell \in [p]$, let $w_\ell(x,y)$ denote the number of directed walks of length exactly $\ell$ from $x$ to $y$ in $\vec{G}$. Let 
		$$ w(x,y) = \sum_{\ell=1}^p w_\ell(x,y). $$ 
		If $w(x,y) > 0$, then the graph $\vec{F}$ will have the edge $(x,y)$ with weight $w(x,y)$, and otherwise (i.e., if $w(x,y) = 0$), it will not have the edge $(x,y)$. 
		Namely, the weight function of $\vec{F}$ is $w_{\vec{F}}(u,v) = w(u,v)$.
		We observe that the graph $\vec{F}$ can be constructed in $O(n)$ time. Indeed, for each $x \in V(\vec{G})$, there are at most $\kappa^\ell$ vertices $y$ such that there is a directed walk of length $\ell$ from $x$ to $y$, as all out-degrees in $\vec{G}$ are at most $\kappa$. Thus, for each $x \in V(\vec{G})$, it takes constant time to compute $w(x,y)$ for all $y$ such that $w(x,y) > 0$. Moreover, we observe that the out-degree of every vertex in $\vec{F}$ is constant. 
		
		Now, let $\varphi: V(\vec{H}) \rightarrow V(\vec{F})$ be a homomorphism. We have that 
		$$
		W(\varphi) = \prod_{(u,v) \in E(\vec{H})} w_{\vec{F}}((\varphi(u), \varphi(v))) = 
		\prod_{(u,v) \in E(\vec{H})} (w_1(\varphi(u), \varphi(v)) + \dots + w_p(\varphi(u), 	\varphi(v))),
		$$
		and so 
		\begin{equation}\label{eq:all-hom-weight-to-sum-weights}
		\hom(\vec{H}, \vec{F}, w_{\vec{F}}) =
		\sum_{\varphi \in \text{Hom}(\vec{H}, \vec{F})} \, \prod_{(u,v) \in E(\vec{H})} (w_1(\varphi(u), \varphi(v)) + \dots + w_p(\varphi(u), 	\varphi(v))).
		\end{equation}	
		Write $e_i = (u_i,v_i)$ for $1 \leq i \leq t$. 
		Crucially, observe that for a digraph $\vec{H}' \in SD_{\vec{H}, p}$ with subdivision sequence $(x_1, \dots, x_t) \in [p]^t$, 
		a homomorphism $\psi: V(\vec{H}') \rightarrow V(\vec{G})$ corresponds to a homomorphism $\varphi: V(\vec{H}) \rightarrow V(\vec{F})$ 
		such that $\varphi(v) = \psi(v)$ for every $v \in V(\vec{H})$, 
		together with a directed walk \linebreak
		$\psi(u_i) \rightarrow \psi(s_{e_i}^1) \rightarrow \dots \rightarrow \psi(s_{e_i}^{x_i-1}) \rightarrow \psi(v_i)$ 
		of length $x_i$ in $\vec{G}$ (where $\{s_{e_i}^j\} \in V(\vec{H}')$ are the subdivision vertices of the path corresponding to the edge $e_i$), for every $1 \leq i \leq t$. 
		Since for each $i$ the number of such walks (in $\vec{G}$) is $w_{x_i}(\varphi(u_i), \varphi(v_i))$, we have that 	
		
		\begin{equation}\label{eq:weight-prod-subdiv-seq}
		\hom(\vec{H}', \vec{G}) = \sum_{\varphi \in \text{Hom}(\vec{H}, \vec{F})} \,
		\prod_{e_i = (u_i, v_i) \in E(\vec{H})} w_{x_i}(\varphi(u_i), \varphi(v_i)).
		\end{equation}
		Recalling that every $\vec{H}' \in SD_{\vec{H}, p}$ is defined by a unique sequence $(x_1, \dots, x_t) \in [p]^t$, by combining \eqref{eq:all-hom-weight-to-sum-weights} and \eqref{eq:weight-prod-subdiv-seq}, we get
		\begin{align}\label{eq:hom-sum-all-in-SD}
		\sum_{\vec{H}' \in SD_{\vec{H}, p}} \hom(\vec{H}', \vec{G}) &=
		\sum_{\varphi \in \text{Hom}(\vec{H}, \vec{F})} \,
		\sum_{(x_1, \dots, x_t) \in [p]^t} \,
		\prod_{e_i = (u_i, v_i) \in E(\vec{H})} w_{x_i}(\varphi(u_i), \varphi(v_i)) \nonumber \\
		&= \sum_{\varphi \in \text{Hom}(\vec{H}, \vec{F})} \, \prod_{e_i = (u_i, v_i) \in E(\vec{H})} (w_1(\varphi(u_i), \varphi(v_i)) + \dots + w_p(\varphi(u_i), 	\varphi(v_i))) \\
		&= \hom(\vec{H}, \vec{F}, w_{\vec{F}}). \nonumber
		\end{align}
		
		We were able to express $\hom(\vec{H}, \vec{F}, w_{\vec{F}})$ as a linear combination of $\hom(\vec{H}', \vec{G})$ (with $\vec{H}' \in SD_{\vec{H}, p}$), where all of the coefficients are equal to 1. We observe that for isomorphic graphs $\vec{H}', \vec{H}'' \in SD_{\vec{H}, p}$, $\hom(\vec{H}', \vec{G}) = \hom(\vec{H}'', \vec{G})$, and thus we can ``combine like terms" in \eqref{eq:hom-sum-all-in-SD}. Namely, let $\vec{H_1}, \dots, \vec{H_k}$ be an enumeration of all graphs in $SD_{\vec{H}, p}$ {\em up to isomorphism} (that is, 
		$\vec{H_1}, \dots, \vec{H_k}$ are pairwise non-isomorphic). Then, there exist positive constants $c_1, \dots, c_k > 0$ such that 
		\begin{equation}\label{eq:hom-sum-non-isomorphic}
		\hom(\vec{H}, \vec{F}, w_{\vec{F}}) = \sum_{i=1}^k c_i \cdot \hom(\vec{H_i}, \vec{G}).
		\end{equation}
		
		As $\vec{F}$ is a weighted digraph where all out-degrees are constant, $\hom(\vec{H}, \vec{F}, w_{\vec{F}})$ can be computed in time $O(n^\alpha)$, by our assumption. Combined with \eqref{eq:hom-sum-non-isomorphic}, we get that $\textsc{hom-cnt}_{c_1 \vec{H_1} + \dots + c_k \vec{H_k}}$ can be solved in time $O(n^\alpha)$ in bounded-degeneracy digraphs.
		
		By Lemma \ref{lem:hom-linear-combination}, $\textsc{hom-cnt}_{\vec{H_i}}$ can be solved in time $O(n^\alpha)$ in degenerate digraphs, for each $1 \leq i \leq k$. 
		This completes the proof, since every digraph 
		$\vec{H}' \in SD_{\vec{H}, p}$ is isomorphic to one of $\vec{H_1}, \dots, \vec{H_k}$.
	\end{proof}
	
	
	\section{Counting Homomorphisms of Cycles: A Combinatorial Algorithm}\label{sec:comb}
	In  this section we prove the second part of Theorem \ref{thm:main}. We prefer to prove this part first as it is technically simpler than the first part.
	Let $P_r$ denote the oriented path with vertices $1,\dots,2r+1$ (so $e(P_r) = 2r$), and where the vertices $1, 3,\ldots,2r+1$ are sinks
	and the vertices $2, 4,\ldots,2r$ are sources. 
	\begin{lemma}\label{lem:path_counting}
		There is an algorithm which, given integers $r,\Delta \geq 1$ and a weighted degenerate $n$-vertex input DAG $\vec{G}$, computes the following:
		\begin{enumerate}
			\item For every $x,y \in V(\vec{G})$, the total weight $N_{r,\Delta}(x,y)$ of homomorphisms $\varphi: P_r \rightarrow \vec{G}$ such that $\varphi(1) = x$, $\varphi(2r+1) = y$, and $\textup{in-deg}(\varphi(t)) \leq \Delta$ for every sink $t$ of $P_r$ other than $1,2r+1$.   
			\item For every $x,y \in V(\vec{G})$ with $\textup{in-deg}(x) > \Delta$ and for every function $f : \{3,5,\dots,2r+1\} \rightarrow \{\textup{low, high}\}$, the total weight $M_{r,\Delta,f}(x,y)$ of homomorphisms $\varphi: P_r \rightarrow \vec{G}$ such that $\varphi(1) = x$ and $\varphi(2r+1) = y$, and 
			such that for every $t \in \{3,5,\dots,2r+1\}$ it holds that $\textup{in-deg}(\varphi(t)) \leq \Delta$ if $f(t) = \textup{low}$ and 
			$\textup{in-deg}(\varphi(t)) > \Delta$ if $f(t) = \textup{high}$.
		\end{enumerate} 
		The computation for Item 1 takes time $\tilde{O}(n\Delta^{r-1})$ and the computation for Item 2 takes time $\tilde{O}(n^2/\Delta)$. 
		Furthermore, the number of pairs $x,y$ with $N_{r,\Delta}(x,y) > 0$ is $O(n\Delta^{r-1})$. 
	\end{lemma}
	\begin{proof}
		We will use a hash map to store the counts $N_{r,\Delta}(x,y)$ and $M_{r,\Delta,f}(x,y)$ (for all pairs $(x,y)$ for which the counts are non-zero). 
		Using such a hash map is the reason for the logarithmic factors implicit in the $\tilde{O}$-notation in the runtime bound, since the hash-map operations require time $\tilde{O}(1)$. 
		This logarithmic factor can be avoided at the cost of allowing randomized algorithms (in which case we should speak of {\em expected} runtime).
		
		Let $w_{\vec{G}}$ be the weight function of $\vec{G}$. 
		We start with Item 1. Here we enumerate all homomorphisms 
		$\varphi: P_r \rightarrow \vec{G}$ such that $\textup{in-deg}(\varphi(t)) \leq \nolinebreak \Delta$ for every sink $t \in V(P_r)$ other than $1,2r+1$. Enumerating all such homomorphisms is clearly sufficient to produce the desired counts. To choose such a homomorphism $\varphi$, we first choose $\varphi(2)$, for which there are $n$ choices. Having chosen $\varphi(2)$, we have $O(1)$ choices for $\varphi(1)$ and $\varphi(3)$, where we only go over the choices for $\varphi(3)$ which satisfy $\textup{in-deg}(\varphi(3)) \leq \Delta$. Having chosen $\varphi(3)$, there are at most $\Delta$ choices for $\varphi(4)$, since $\varphi(4)$ must be an in-neighbor of $\varphi(3)$. Continuing in this fashion, we see that there are $O(1)$ choices for each of the vertices $\varphi(1),\varphi(3),\dots,\varphi(2r+1)$ and at most $\Delta$ choices for each of the vertices $\varphi(4),\varphi(6),\dots,\varphi(2r)$. Hence, the total number of choices is $n\Delta^{r-1}$. In particular, the number of pairs $x,y$ with $N_{r,\Delta}(x,y) > 0$ (namely, which are the endpoints in such a homomorphism) is at most $O(n\Delta^{r-1})$, as required.
		
		We now establish Item 2, whose proof is by induction on $r$. 
		First, let us handle the base case $r=1$. Observe that $P_1$ is just the two-edge oriented path with the middle vertex $2$ being a source and the two endpoints $1,3$ being sinks. We can enumerate all homomorphisms $\varphi$ from $P_1$ to $\vec{G}$ in time $O(n)$, by going over all (at most $n$) choices for $\varphi(2)$ and for each such choice, going over all (at most $O(1)$) pairs of out-neighbors of $\varphi(2)$. Enumerating all homomorphisms of $P_1$ is clearly enough to compute the desired counts.
		
		We now proceed to the induction step. Assume, by the induction hypothesis, that we have already computed the desired counts for $P_{r-1}$. To achieve this for $P_r$, go over all choices for vertices $x,u \in V(\vec{G})$ such that $\textup{in-deg}(x) > \Delta$. The number of such choices is at most $O(n^2/\Delta)$, since there are at most $n$ choices for $u$ and at most $O(n/\Delta)$ choices for $x$. Now go over all pairs $v,y$ of out-neighbors of $u$. For each such pair $v,y$, and for each function 
		$g : \{3,5,\dots,2r-1\} \rightarrow \{\textup{low, high}\}$, we have already computed the total weight $M_{r-1,\Delta,g}(x,v)$ of homomorphisms $\psi : P_{r-1} \rightarrow \vec{G}$ such that $\psi(1) = x$ and $\psi(2r-1) = v$, and such that for every $t \in \{3,5,\dots,2r-1\}$ it holds that $\textup{in-deg}(\varphi(t)) \leq \Delta$ if $g(t) = \textup{low}$ and 
		$\textup{in-deg}(\varphi(t)) > \Delta$ if $g(t) = \textup{high}$.
		Now, for each such $\psi$ and $g$, the map $\varphi : V(P_r) \rightarrow V(\vec{G})$ which agrees with $\psi$ on $P_r[\{1,\dots,2r-1\}] = P_{r-1}$ and satisfies $\varphi(2r) = u$ and $\varphi(2r+1) = y$, is a homomorphism from $P_r$ to $\vec{G}$ satisfying $\varphi(1) = x$ and $\varphi(2r+1) = y$. Furthermore, the function 
		$f : \{3,5,\dots,2r+1\} \rightarrow \{\textup{low, high}\}$ corresponding to $\varphi$ is simply the function which agrees with $g$ on $\{3,5,\dots,2r-1\}$ and satisfies $f(2r+1) = \textup{low}$ if $\textup{in-deg}(y) \leq \Delta$ and $f(2r+1) = \textup{high}$ if $\textup{in-deg}(y) > \Delta$. 
		Hence, multiplying the weight of $\psi$ by $w(u,v) \cdot w(u,y)$ and summing over all such $\psi$, we obtain the desired count for the pair $x,y$. 
	\end{proof}
	\noindent
	Let $\vec{C}_{2\ell}$ denote the alternating orientation of $C_{2\ell}$.
	\begin{theorem}\label{thm:comb}
		For every $\ell \geq 2$, there is an algorithm which, given a weighted degenerate $n$-vertex input DAG $\vec{G}$, computes the total weight of homomorphisms from $\vec{C}_{2\ell}$ to $\vec{G}$ in time $\tilde{O}(n^{2 - 1/\lceil \ell/2 \rceil})$. 
	\end{theorem}
	\begin{proof}
		Set $\Delta := n^{1/\lceil \ell/2 \rceil}$. 
		Denote the sinks of $\vec{C}_{2\ell}$ by $t_1,\dots,t_{\ell}$. In order to compute $\hom(\vec{C}_{2\ell},\vec{G},w_{\vec{G}})$, we will compute, for every function 
		$g : \{t_1,\dots,t_{\ell}\} \rightarrow \{\textup{low, high}\}$, the total weight $\hom_g$ of homomorphisms 
		$\varphi : \vec{C}_{2\ell} \rightarrow \vec{G}$ such that $\textup{in-deg}(\varphi(t_i)) \leq \Delta$ if $g(t_i) = \textup{low}$ and 
		$\textup{in-deg}(\varphi(t_i)) > \Delta$ if $g(t_i) = \textup{high}$ (for every $1 \leq i \leq \ell$). Clearly, $\hom(\vec{C}_{2\ell},\vec{G},w_{\vec{G}}) = \sum_{g}{\hom_g}$. We consider two cases. Suppose first that $g(t_i) = \textup{low}$ for every $1 \leq i \leq \ell$. It is easy to see that $\vec{C}_{2\ell}$ consists of a copy of $P_{\lfloor \ell/2 \rfloor}$ and a copy of $P_{\lceil \ell/2 \rceil}$, glued together at their endpoints. It is now easy to see that for this particular $g$, one has 
		$$
		\hom_g = \sum_{\substack{x,y \in V(\vec{G}): \\ \textup{in-deg}(x) \leq \Delta \\ \textup{in-deg}(y) \leq \Delta}}
		{N_{\lfloor \ell/2 \rfloor,\Delta}(x,y) \cdot 
			N_{\lceil \ell/2 \rceil,\Delta}(x,y)},
		$$
		where $N_{r,\Delta}(x,y)$ is as defined in Item 1 of Lemma \ref{lem:path_counting}. Since we can compute the counts 
		\linebreak $N_{\lfloor \ell/2 \rfloor,\Delta}(x,y), 
		N_{\lceil \ell/2 \rceil,\Delta}(x,y)$ for all pairs of vertices $x,y \in V(G)$ (simultaneously) in time $\tilde{O}(n \Delta^{\lceil \ell/2 \rceil-1})$, and since the total number of pairs for which these counts are non-zero is 
		$O(n \Delta^{\lceil \ell/2 \rceil-1})$, the above sum can be computed in time 
		$\tilde{O}(n \Delta^{\lceil \ell/2 \rceil-1}) = 
		\tilde{O}(n^{2 - 1/\lceil \ell/2 \rceil})$, as required.
		
		Suppose now that $g(t_i) = \textup{high}$ for some $1 \leq i \leq \ell$.
		We proceed similarly to the previous case. Decompose $\vec{C}_{2\ell}$ into a copy $P$ of $P_{\lfloor \ell/2 \rfloor}$ and a copy $P'$ of $P_{\lceil \ell/2 \rceil}$, such that $t_i$ is an endpoint of both of these paths. Let $f,f'$ be the ``low/high signatures" corresponding the paths $P,P'$, respectively. Again, it is not hard to see that 
		$$
		\hom_g = 
		\sum_{\substack{x,y \in V(\vec{G}): \\ \textup{in-deg}(x) > \Delta}}
		{M_{\lfloor \ell/2 \rfloor,\Delta,f}(x,y) \cdot 
			M_{\lceil \ell/2 \rceil,\Delta,f'}(x,y) }.
		$$
		By Item 2 of Lemma \ref{lem:path_counting}, one can compute $M_{\lfloor \ell/2 \rfloor,\Delta,f}(x,y)$ and
		$M_{\lceil \ell/2 \rceil,\Delta,f'}(x,y)$ for all pairs 
		$x,y \in V(\vec{G})$ with $\textup{in-deg}(x) > \Delta$ in time $\tilde{O}(n^2/\Delta)$. Furthermore, the number of such pairs $x,y$ is $O(n^2/\Delta)$, because there are $O(n/\Delta)$ choices for $x$. We conclude that the above sum can be computed in time 
		$\tilde{O}(n^2/\Delta) = \tilde{O}(n^{2 - 1/\lceil \ell/2 \rceil})$, as required. This completes the proof of the theorem. 
	\end{proof}
	
	\begin{proof}[Proof of Theorem \ref{thm:main}, Part (ii)]
		Let $G$ be an $n$-vertex degenerate graph and consider the DAG $\vec{G}$ obtained (in linear time) from a degenerate ordering of $G$. We need to compute $\hom(C_{2k},G)$ in $O(n^{2-1/\lceil k/2 \rceil})$ time (the proof for $\hom(C_{2k+1},G)$ is analogous). Recall that this amounts to
		computing $\hom(\vec{H},\vec{G})$ for every acyclic orientation $\vec{H}$ of $C_{2k}$.
		So, consider some such $\vec{H}$. If $\vec{H}$ has exactly one source, then we can enumerate all homomorphisms from $\vec{H}$ to $\vec{G}$ --- and in particular compute $\hom(\vec{H},\vec{G})$ --- in linear time\footnote{This is because after choosing the image of the unique source of $\vec{H}$ (under a homomorphism from $\vec{H}$ to $\vec{G}$), there are $O(1)$ choices for every other vertex of $\vec{H}$, as all vertices of $\vec{H}$ are reachable from the source and $\vec{G}$ has bounded out-degrees.}; see, e.g., \cite{Bressan}. 
		Suppose then that $\vec{H}$ has at least two sources, and recall that $\vec{H}$ is a directed subdivision of the alternating
		orientation $\vec{C}_{2\ell}$ of $C_{2\ell}$ for some $2 \le \ell \le k$.
		It now follows from Theorem \ref{thm:reduction-subdivision_main} that it suffices to show that
		$\textsc{w-hom-cnt}_{\vec{C}_{2\ell}}$ can be solved in time $\tilde{O}(n^{2-1/\lceil k/2 \rceil})$
		in weighted degenerate DAGs, which is precisely the statement of Theorem \ref{thm:comb}.
	\end{proof}
	
	\section{Counting Homomorphisms of Cycles: Using Matrix Multiplication}\label{sec:mat-mult}
	
	In  this section we prove the first part of Theorem \ref{thm:main}.
	%
	%
	We begin with describing an algorithm for counting weighted homomorphisms of alternating orientations of cycles. As mentioned in the introduction, our algorithm is inspired by the algorithm of Yuster and Zwick \cite{YZ} for finding a directed $k$-cycle in a (general) directed graph but bears several significant differences and tweaks.
	As in \cite{YZ}, our approach for enumerating (weighted) homomorphisms is by classifying them into a polylogarithmic number of classes, according to the in-degrees of the images of the (sink) vertices of the $\vec{C}_{2k}$ (the alternating orientation of $C_{2k}$).
	
	\begin{theorem}\label{thm:UZ-C_2k-alternating}
		Let $k \geq 2$ and let $\vec{C}_{2k}$ be the alternating orientation of $C_{2k}$. Then, the
		$\textsc{w-hom-cnt}_{\vec{C}_{2k}}$ problem can be solved in degenerate digraphs in time $\tilde{O}(n^{c_k})$.
	\end{theorem}
	
	\begin{proof}
		Let $\vec{G}$ be an $n$-vertex degenerate weighted digraph with weight function $w_{\vec{G}}$. Let us denote by $u_0, \dots, u_{k-1}$ the source vertices of $\vec{C}_{2k}$, and by $s_0, \dots, s_{k-1}$ the sink vertices of $\vec{C}_{2k}$, such that $(u_i,s_i), (u_i, s_{i-1})$ are edges for all $i$ (here and in what follows, indices are taken modulo $k$).  
		
		We call a triple of vertices $x,y,z \in V(\vec{G})$ such that $(z,x), (z,y)$ are edges a {\em cherry}. We start by observing that in $O(n)$ time one can enumerate all the cherries in $\vec{G}$, as the out-degree of every vertex in $\vec{G}$ is bounded. In particular, for every pair of vertices $x,y \in V(\vec{G})$, we will store in a hash table, denoted $\mathcal{HM}$, the sum of weight products of cherries with endpoints $x,y$ (i.e., the sum of $w_{\vec{G}}(z,x) \cdot w_{\vec{G}}(z,y)$ over all vertices $z$ such that $(z,x), (z,y)$ are edges). This number will be denoted by $\mathcal{HM}[x,y]$. Note that there are only $O(n)$ pairs $(x,y)$ for which $\mathcal{HM}[x,y] > 0$.
		
		Let us partition the vertices of $\vec{G}$ into $\log n$ degree classes, as follows. For every $0 \leq i < \log n $, let
		$$ W_i = \{v \in V(\vec{G}) \, | \, 2^i \leq \text{in-deg}(v) < 2^{i+1}\}. $$ 
		We refer to a degree class $W_j$ by its index $j$, for simplicity of notation. 
		Clearly, we have that $|W_i| = O(n/2^i)$ (recall that $\vec{G}$ has $O(n)$ edges). 
		
		Our approach for enumerating (weighted) homomorphisms from $\vec{C}_{2k}$ to $\vec{G}$ will be classifying them into $O(\log^k n)$ classes, according to the degree classes of the images of $s_0, \dots, s_{k-1}$.
		
		Now, let us fix a tuple of degree classes $f = (f_0, \dots, f_{k-1})$ with $f_r \in \{0, \dots, \log n\}$. Our goal is to compute the weighted homomorphism count of homomorphisms $\varphi: \vec{C}_{2k} \rightarrow \vec{G}$ such that $\varphi(s_j) \in W_{f_j}$ for all $0 \leq j < k$ (i.e., $\varphi(s_j)$ has in-degree roughly $2^{f_j}$). 
		For $i,j \in \{0, \dots, \log n\}$, we denote by $A_{i,j}$ the $|W_i| \times |W_j|$ matrix such that for every $x \in W_i, y \in W_j$, $A_{i,j}(x,y) = \mathcal{HM}[x,y]$. We will sometimes represent $A_{i,j}$ as a {\em sparse} matrix, e.g., using adjacency lists. 
		Note that $A_{i,j}$ has only $O(n)$ non-zero entries, since there are in total only $O(n)$ pairs $x,y \in V(\vec{G})$ for which $\mathcal{HM}[x,y] > 0$. Hence, a sparse representation of $A_{i,j}$ can be obtained in time $O(n)$. 		
		
		Now, for $p,q \in \{0, \dots, k-1\}$, we let 
		$$ B^f_{p,q} = A_{f_p, f_{p+1}} A_{f_{p+1}, f_{p+2}} \cdots  A_{f_{q-1}, f_q}. $$
		Observe that for a given $p \in \{0, \dots, k-1\}$, $A_{f_p, f_{p+1}}(x,y)$ corresponds to the number of homomorphisms $\varphi$ of $\vec{C}_{2k}[\{s_p, u_{p+1}, s_{p+1} \}]$ such that $\varphi(s_p) = x$, $\varphi(s_{p+1}) = y$. We thus get that $B^f_{p,q}(x,y)$ corresponds to the number of homomorphisms $\varphi$ of $\vec{C}_{2k}[\{s_p, u_{p+1}, s_{p+1}, \dots, u_q, s_q \}]$ such that $\varphi(s_p) = x$, $\varphi(s_q) = y$, and $\varphi(s_j) \in W_{f_j}$ for all $p < j < q$. (This is a variation of the standard way of counting $r$-walks in a digraph using the $r^{\text{th}}$ power of its adjacency matrix.)
		Now, one can observe that the weighted homomorphism count of homomorphisms $\varphi: \vec{C}_{2k} \rightarrow \vec{G}$ such that $\varphi(s_j) \in W_{f_j}$ for $0 \leq j < k$ can be obtained in the following way. We consider the matrix chain product 
		$B^f_{0,k} = A_{f_0, f_1} A_{f_1, f_2} \cdots  A_{f_{k-2}, f_{k-1}} A_{f_{k-1}, f_0}$.
		The sum of the diagonal entries of $B^f_{0,k}$ (i.e., $\text{trace}(B^f_{0,k})$) gives the desired weighted homomorphism count. In order to compute $\text{trace}(B^f_{0,k})$ efficiently, our approach will be to pick a particular pair $i, j \in \{0, \dots, k-1\}$, compute $B^f_{i,j}$ and $B^f_{j,i}$, and then use them to compute $\text{trace}(B^f_{i,j} B^f_{j,i}) = \text{trace}(B^f_{0,k})$. Our choice of $i,j$ will be such that the running time is minimized. 
		
		It will be more convenient to express the degrees in terms of the number of edges in $\vec{G}$, which is $O(n)$. Specifically, when considering a $k$-tuple of degree classes $(f_0, \dots, f_{k-1})$ with $f_r \in \{0, \dots, \log n\}$, we will let $n^{d_j} = 2^{f_j}$, and so $d_j = f_j / \log n$. Notice that $0 \leq d_j \leq 1$, and that $|W_j| = O(n^{1-d_j})$.
		
		For a fixed tuple of degree classes $d = (d_0, \dots, d_{k-1})$, let $P^d_{i,j}$ be the minimum such that $B^d_{i,j}$ can be computed in time $O(n^{P^d_{i,j}})$. We observe that this also gives an upper bound on the number of non-zero entries in $B^d_{i,j}$. 
		The matrix $B^d_{i,j}$ can be computed in three ways, as follows: 
		
		\begin{enumerate}
			\item Compute a sparse representation of $B^d_{i,j-1}$. Then, for every entry $(x,y) \in W_{f_i} \times W_{f_{j-1}}$ of $B^d_{i,j-1}$, traverse all of the in-neighbors of $y$, and
			denote by $S_y$ the set of their out-neighbors in $W_{f_j}$ (i.e., $w \in S_y$ if and only if $w \in W_{f_j}$ and there exists $u \in V(\vec{G})$ such that $(u,y),(u,w) \in E(\vec{G})$). Now, for each $w \in S_y$, update
			$B^d_{i,j}(x,w) = B^d_{i,j}(x,w) + B^d_{i,j-1}(x,y) \cdot \mathcal{HM}[y,w]$. 
			The computation of $B^d_{i,j-1}$ takes $O(n^{P^d_{i,j-1}})$ time. Now, each $y \in W_{f_{j-1}}$ has at most $O(n^{d_{j-1}})$ in-neighbors, with each having $O(1)$ out-neighbors (i.e., $|S_y|=O(n^{d_{j-1}})$ for each $y \in W_{f_{j-1}}$). Therefore, a sparse representation of $B^d_{i,j}$ can be computed in time $O(n^{P^d_{i,j-1} + d_{j-1}})$ (as the size of a sparse representation of $B^d_{i,j-1}$ is $O(n^{P^d_{i,j-1}})$).
			
			\item Similar to the above, but reversing the roles of $j-1$ and $i$: Compute a sparse representation of $B^d_{i+1,j}$. Then, for every entry $(x,y) \in W_{f_{i+1}} \times W_{f_j}$ of $B^d_{i+1,j}$, traverse all of the in-neighbors of $x$, and denote by $S_x$ the set of their out-neighbors in $W_{f_i}$ (i.e., $w \in S_x$ if and only if $w \in W_{f_i}$ and there exists $u \in V(\vec{G})$ such that $(u,x),(u,w) \in E(\vec{G})$).
			Now, for each $w \in S_x$, update
			$B^d_{i,j}(w,y) = B^d_{i,j}(w,y) + \mathcal{HM}[w,x] \cdot B^d_{i+1,j}(x,y)$. 
			The computation of $B^d_{i+1,j}$ takes $O(n^{P^d_{i+1,j}})$ time. Now, each $x \in W_{f_{i+1}}$ has at most $O(n^{d_{i+1}})$ in-neighbors, with each having $O(1)$ out-neighbors (i.e., $|S_x|=O(n^{d_{i+1}})$ for each $x \in W_{f_{i+1}}$). Therefore, a sparse representation of $B^d_{i,j}$ can be computed in time $O(n^{P^d_{i+1,j} + d_{i+1}})$ (as the size of a sparse representation of $B^d_{i+1,j}$ is $O(n^{P^d_{i+1,j}})$). 
			
			\item For some $i < r < j$, compute $B^d_{i,r}$ and $B^d_{r,j}$. Then, compute their product to obtain $B^d_{i,j}$. It takes $O(n^{P^d_{i,r}} + n^{1-d_i}n^{1-d_r})$ time to compute the (non-sparse) matrix representation of $B^d_{i,r}$, and $O(n^{P^d_{r,j}} + n^{1-d_r}n^{1-d_j})$ time to compute the (non-sparse) matrix representation of $B^d_{r,j}$.
			Finally, it takes $O(n^{M(1-d_i, 1-d_r, 1-d_j)})$ time to compute their product, where $M(a,b,c)$ is the smallest $g$ such that one can multiply an $n^a \times n^b$ by an $n^b \times n^c$ matrix in time $O(n^g)$. It is not difficult to see (see, e.g., \cite{HP}), that $M(a,b,c) \leq a+b+c - (3-\omega)\min\{a,b,c\}$.
		\end{enumerate}
		
		Now, the exponent of the running time in Item 1 is recursively bounded as 
		$P^d_{i,j} \leq P^d_{i,j-1} + d_{j-1}$. Similarly, for Item 2 we have the bound 
		$P^d_{i,j} \leq P^d_{i+1,j} + d_{i+1}$. On the other hand, the running time of Item 3 is bounded by  
		$$ P^d_{i,j} \leq \min_{i < r < j} \max \{P^d_{i,r}, P^d_{r,j}, M(1-d_i, 1-d_r, 1-d_j)\}. $$
		
		We observe that $P^d_{i,i+1} = 1$ as a sparse representation of $B^d_{i,i+1} = A_{f_{i},f_{i+1}}$ can be obtained in time $O(n)$. 
		To summarize, we have the following inductive definition:
		\begin{equation}\label{eq:runtime-recur-Pij}
		\begin{gathered}
		P^d_{i,i+1} = 1 \\
		P^d_{i,j} = \min \{P^d_{i,j-1} + d_{j-1}, P^d_{i+1,j} + d_{i+1}, 
		\min_{i < r < j} \max \{P^d_{i,r}, P^d_{r,j}, M(1-d_i, 1-d_r, 1-d_j)\} \}
		\end{gathered}
		\end{equation}
		
		Now, given sparse representations of $B^d_{i,j}$ and $B^d_{j,i}$ for some $i, j \in \{0, \dots, k-1\}$, we can compute $\text{trace}(B^d_{i,j} B^d_{j,i})$ in time of the number of non-zero entries in $B^d_{i,j}$ and $B^d_{j,i}$, which is $O(n^{P^d_{i,j}} + n^{P^d_{j,i}})$. This is true as $\text{trace}(B^d_{i,j} B^d_{j,i}) = \sum_{p,q} (B^d_{i,j})_{p,q} \cdot (B^d_{j,i})_{q,p}$, and so we can exploit the sparse representations of $B^d_{i,j}$ and $B^d_{j,i}$\footnote{One way to do it is as follows. For every entry $(y,x)$ in the sparse representation of $B^d_{j,i}$, store its value in a hash table, with $(y,x)$ being the key. Then, go over all entries $(x,y)$ in the sparse representation of $B^d_{i,j}$, and check whether the key $(y,x)$ is in the hash table. If so, accumulate the product of the two entries.}. We note that if  $B^d_{i,j}$ and $B^d_{j,i}$ are given in (non-sparse) matrix representations, we can always convert them into spare representations, as the running time of this conversion is dominated by the time it took to create those matrices.  
		Finally, recall that $\text{trace}(B^d_{i,j} B^d_{j,i})$ corresponds to the weighted homomorphism count for the current degree class $d$.
		
		For $d = (d_0, \dots, d_{k-1})$, define
		\begin{equation}\label{eq:big-C(K)}
		C_k(d_0, \dots, d_{k-1}) = \min_{0 \leq i < j \leq k-1} \max \{P^d_{i,j}, P^d_{j,i}\}.
		\end{equation}
		
		Now, given $d = (d_0, \dots, d_{k-1})$, the algorithm solves a dynamic programming problem of constant size, based on \eqref{eq:runtime-recur-Pij} and \eqref{eq:big-C(K)}, and determines the optimal way of computing the weighted homomorphism count of the current degree class. This is then repeated for all degree classes. 
		
		Therefore, the running time for computing the weighted homomorphism count for a degree class $d$ is $O(n^{C_k(d_0, \dots, d_{k-1})})$, and so the total running time is $\tilde{O}(n^{c_k})$, where 
		\begin{equation}\label{eq:small-c_k}
		c_k = \max_{d=(d_0, \dots, d_{k-1})} C_k(d_0, \dots, d_{k-1})\footnote{The definition of $c_k$ in \cite{DVW} (which corresponds to the one of Yuster and Zwick) has an additional component (which comes from a certain combinatorial procedure, that is not applicable to our algorithm). More specifically, it is defined as 
			$c_k = \max_{d=(d_0, \dots, d_{k-1})} \min \left\{ \min_{0 \le i \le k-1} (2-d_i), 
			C_k(d_0, \dots, d_{k-1}) \right\}$.
			Seemingly, this definition is not equivalent to ours, and (possibly) gives a better running time. However, this is not the case. In particular, the upper bounds on $c_k$ (as defined in \cite{DVW}) are obtained by attaining upper bounds on $C_k$. Specifically, it is shown in \cite{DVW}, that for odd $k$ and arbitrary $d_0, \dots, d_{k-1}$, $C_k(d_0, \dots, d_{k−1}) \le \frac{(k+1)\omega}{2\omega+k-1}$,
			and, for $\omega \le \frac{2k}{k-1}$, a specific sequence $d=(d_0, \dots, d_{k−1})$ is presented such that $\min_{0 \le i \le k-1} (2-d_i) \ge C_k(d_0, \dots, d_{k−1}) = \frac{(k+1)\omega}{2\omega+k-1}$ (see Section 5.4 in \cite{DVW});
			and for even $k$ and arbitrary $d_0, \dots, d_{k-1}$, 
			$C_k(d_0, \dots, d_{k−1}) \le \frac{k\omega - \frac{4}{k}}{2\omega + k-2-\frac{4}{k}}$, 
			and, for $\omega = 2$, a specific sequence $d=(d_0, \dots, d_{k−1})$ is presented such that $\min_{0 \le i \le k-1} (2-d_i) \ge C_k(d_0, \dots, d_{k−1}) = \frac{k\omega - \frac{4}{k}}{2\omega + k-2-\frac{4}{k}}$ 
			(see Section 5.6 in \cite{DVW}). Therefore, as the minimum between $\min_{0 \le i \le k-1} (2-d_i)$ and $C_k(d_0, \dots, d_{k-1})$ is chosen, we conclude that the expression in our definition of $c_k$ is, in fact, equal to the one presented in \cite{DVW}.}.
		\end{equation}
		
		Both \eqref{eq:runtime-recur-Pij} and \eqref{eq:big-C(K)} are defined by Dalirrooyfard et al. (see Section 5 in \cite{DVW}) in the exact same manner, as part of their description and analysis of the Yuster-Zwick algorithm \cite{YZ}.
		
		We conclude that the $\textsc{w-hom-cnt}_{\vec{C}_{2k}}$ problem can be solved in degenerate graphs in time $\tilde{O}(n^{c_k})$.	
	\end{proof}
	
	%
	%
	
	\begin{proof}[Proof of Theorem \ref{thm:main}, Part (i)]
		This is analogous to the proof of Part (ii), except that now we use Theorem \ref{thm:UZ-C_2k-alternating}.
		So, let $G$ be an $n$-vertex degenerate graph and let $\vec{G}$ be the DAG corresponding to a degenerate ordering of $G$. We need to compute $\hom(C_{2k},G)$ in $\tilde{O}(n^{c_k})$ time (the proof for $\hom(C_{2k+1},G)$ is analogous). As in the proof of Part (ii), it suffices to show that for $2 \le \ell \le k$,
		$\textsc{w-hom-cnt}_{\vec{C}_{2\ell}}$ can be solved in time $\tilde{O}(n^{c_k})$
		in weighted degenerate digraphs. 
		And this is indeed the case by Theorem \ref{thm:UZ-C_2k-alternating}.
	\end{proof}	
	
	\section{Counting Cycle Homomorphisms vs. Directed-Cycle Detection}\label{sec:count-vs-detect}
	
	
	%
	
	In this section we prove Theorems \ref{thm:count-vs-detect} and \ref{thm:detect-degeneracy}
	as they bear some similarity. 
	An important ingredient in the proof of Theorem \ref{thm:count-vs-detect} is the fact that the number of {\em cycle transversals} in partite graphs can be computed efficiently using an inclusion-exclusion procedure. 
	To make this statement precise we need the following definition.
	\begin{definition}\label{def:transversal}
		Let $G$ be a $p$-partite graph with a given vertex partition 
		${\mathcal P}=\{V_1,\ldots,V_p\}$.
		A {\em ${\mathcal P}$-cycle transversal} of $G$ is a simple cycle of length $p$ in $G$ containing a single vertex from each part of ${\mathcal P}$.
	\end{definition}
	
	A family of graphs ${\mathcal F}$ is {\em subgraph-closed} if $G \in {\mathcal F}$ implies that every subgraph of $G$ is also in ${\mathcal F}$. For example, $d$-degenerate graphs are subgraph-closed.
	\begin{lemma}\label{lem:transversal}\footnote{A similar lemma appears in \cite{C_Marx}.}
		Let $p \ge 3$. Suppose that ${\mathcal F}$ is a subgraph-closed family of graphs and that $\textsc{hom-cnt}_{C_p}$ can be computed in $f(n)$ time for graphs in ${\mathcal F}$.
		Then, given a $p$-partite $n$-vertex graph $G \in {\mathcal F}$ together with a vertex partition ${\mathcal P}$,
		the number of ${\mathcal P}$-cycle transversals of $G$ can be computed in $O(f(n))$ time.
	\end{lemma}
	
	\begin{proof}
		Let ${\mathcal P}=\{V_1,\ldots,V_p\}$. For a non-empty subset $S \subseteq [p]$, let $G_S$ denote the
		$|S|$-partite induced subgraph of $G$ with partition ${\mathcal P}_S=\{V_i \,|\, i \in S\}$.
		Let $M$ denote the number of homomorphisms that are ${\mathcal P}$-cycle transversals of $G$ (the number of ${\mathcal P}$-cycle transversals can be derived from the number of such homomorphisms after dividing by $2p$).
		By the inclusion-exclusion principle we have
		$$
		M = \sum_{S \subseteq [p]} (-1)^{p-|S|}\hom(C_p,G_S)\;.
		$$
		Since ${\mathcal F}$ is subgraph-closed and since $G_S$ has at most $n$ vertices, we can compute $\hom(C_p,G_S)$ in $f(n)$ time. Hence, $M$ can be computed in $O(2^p f(n))=O(f(n))$ time, as required.
	\end{proof}
	
	\begin{proof}[Proof of Theorem \ref{thm:count-vs-detect}]
		We assume that there is an algorithm that computes $\textsc{hom-cnt}_{C_{2k}}$ in \linebreak $2$-degenerate graphs in $O(n^\alpha)$ time. Let $G$ be an arbitrary directed graph with $m$ edges. Recall that our goal is to decide if $G$ has a simple $C_k$. We first describe a randomized algorithm, based on the color-coding method \cite{AYZ_ColorCoding}. Randomly partition the vertices of $G$ into $k$ parts $V_1,\ldots,V_k$. Let $G'$ be the subgraph of $G$ consisting only of the edges $(u,v)$ such that there exists $1 \le i \le k$ where $u \in V_i$ and $v \in V_{i+1}$ (indices taken modulo $k$).
		Notice that $G'$ has no directed cycles whose length is shorter than $k$, and with constant probability (at least $k^{-k}$), $G'$ has a directed $C_k$ if $G$ has one.
		
		Next, subdivide each edge $(u,v)$ of $G'$ into two edges $(u,z_{uv})$ and $(z_{uv},v)$ where $z_{uv}$ is a new vertex, and denote the newly obtained graph by $G^*$. Furthermore, ignore edge directions in $G^*$
		and hence $G^*$ is an undirected $2$-degenerate graph with $O(m)$ vertices. Also,
		$G^*$ is $(2k)$-partite as can be seen from the partition of
		of its vertex set ${\mathcal P}=\{V_1,V_{1,2},V_2,V_{2,3},\ldots,V_k,V_{k,1}\}$ where
		$V_{i,i+1} = \{z_{uv}\,|\, u \in V_i,\, v \in V_{i+1}\,, (u,v) \in E(G')\}$.
		The simple but crucial point now is that $G'$ has a directed $C_k$ if and only if $G^*$
		has a ${\mathcal P}$-cycle transversal. Indeed, every ${\mathcal P}$-cycle transversal must be of the form
		$(u_1,z_{u_1u_2},u_2,z_{u_2u_3},\ldots,u_k,z_{u_k u_1})$ where $(u_1,\ldots,u_k)$ is a directed $C_k$ in $G$.
		Now, by Lemma \ref{lem:transversal}, we can decide if $G^*$ has a ${\mathcal P}$-cycle transversal in
		$O(|v(G^*)|^\alpha)=O(m^\alpha)$ time, implying that we can decide if $G'$ has a $C_k$ in the same time.
		
		To make the algorithm deterministic we recall that the color-coding method can be derandomized. In particular, it is shown in \cite{AYZ_ColorCoding} that one can deterministically compute $O(\log m)$ partitions of the vertices of $G$ (the partitions can be computed in $O(m \log m)$ time), such that if $G$ has a $C_k$, then for at least
		one of the partitions, the corresponding $G'$ has a $C_k$. Hence, the (now deterministic) running time to
		decide if $G$ has a $C_k$ is $O(m\log m + m^\alpha\log m )=O(m^\alpha \log m)$ as clearly we can assume that \nolinebreak $\alpha \ge 1$.
		
		Observe that in the statement of Theorem \ref{thm:count-vs-detect} we also claim a similar result for $C_{2k+1}$, not just $C_{2k}$ as we have just proved. Indeed,
		it is straightforward to modify the proof of Theorem \ref{thm:count-vs-detect} to obtain the following slightly more general statement: Let $k \ge 3$ and let $p \ge 2k$. If there is an algorithm that computes $\textsc{hom-cnt}_{C_p}$ in $2$-degenerate graphs in $O(n^\alpha)$ time, then there is an algorithm that decides if an arbitrary directed $m$-edge graph has a $C_k$ in time $O(m^\alpha)$ (randomized) and $O(m^\alpha \log m)$ (deterministic). One simple way to achieve this is to subdivide each edge $(u,v)$ in $G'$ into a path of length $2$ as in the proof of  Theorem \ref{thm:count-vs-detect}, except for edges of the form $(u,v)$ where $u \in V_1$ and $v \in V_2$ which are subdivided into a path of length $p+2-2k$.
	\end{proof}
	
	\begin{proof}[Proof of Theorem \ref{thm:detect-degeneracy}]
		We prove for undirected degenerate graphs; the proof for directed degenerate graphs is similar.
		Let $k \ge 6$. Let $G$ be an undirected degenerate graph with $n$ vertices in which we wish to detect a $C_k$. As in the proof of Theorem \ref{thm:count-vs-detect} we use color-coding. Randomly partition the vertices of $G$ into $k$ parts $V_1,\ldots,V_k$. Let $G'$ be the subgraph of $G$ consisting only of the edges $(u,v)$ such that there exists $1 \le i \le k$ where $u \in V_i$ and $v \in V_{i+1}$ (indices taken modulo $k$).
		Notice that $G'$ has no cycle transversal, unless $G$ has a $C_k$, and conversely, if $G$ has a $C_k$ then $G'$ has a cycle transversal with constant positive probability.
		By Lemma \ref{lem:transversal} and by Theorem \ref{thm:main}, we can determine if $G'$ has a cycle transversal in $\tilde{O}(n^{d_{\lfloor k/2 \rfloor}})$ time.
		Again, as color-coding can be derandomized, we can also obtain a deterministic algorithm in
		$\tilde{O}(n^{d_{\lfloor k/2 \rfloor}})$ time.
	\end{proof}
	
	\section{A Combinatorial Algorithm for Counting Cycle Homomorphisms in General Graphs}\label{sec:general_graphs}
	
	As our final result, we use our methods to obtain 
	a combinatorial algorithm for counting the number of $C_k$-homomorphisms in a (general) directed or undirected graph with $m$ edges in $\tilde{O}(m^{2-1/\lceil k/2 \rceil})$ time:
	\begin{theorem}\label{thm:hom_cnt_general_graphs}
		For every $k \geq 3$, there is an algorithm which, given an input (di)graph $G$ with $m$ edges, computes the number of homomorphisms from $C_k$ to $G$ in time $\tilde{O}(m^{2 - 1/\lceil k/2 \rceil})$. 
	\end{theorem}
	\noindent
	In the directed case, the $C_k$ above denotes the directed $k$-cycle. 
	
	%
	
	It should be pointed out that, as in the proof of Theorem \ref{thm:detect-degeneracy}, we can use Theorem \ref{thm:hom_cnt_general_graphs} to obtain an algorithm which detects if a directed or undirected graph has a $C_k$ in $\tilde{O}(m^{2-1/\lceil k/2 \rceil})$ time. The proof in the directed setting is particularly simple: given a digraph $\vec{G}$, we randomly partition its vertices into sets $V_1,\dots,V_k$ and only keep the edges going from $V_i$ to $V_{i+1}$ (for some $1 \leq i \leq k)$, denoting the resulting subdigraph by $\vec{G}'$. 
	If $\vec{G}$ contains a $C_k$ then with constant positive probability so does $\vec{G}'$. Moreover, every $C_k$-homomorphism in $\vec{G}'$ corresponds to a proper copy of $C_k$, so computing $\hom(C_k,\vec{G}')$ suffices in order to decide if $\vec{G}'$ contains a $C_k$. Again, this can be derandomized at the cost of increasing the runtime by a logarithmic \nolinebreak factor.

	
	
	%
	

	The proof of Theorem \ref{thm:hom_cnt_general_graphs} is similar to that of Theorem \ref{thm:comb}. To keep the presentation simple, we focus on undirected graphs; the proof for digraphs is similar. 
	Let $P_r$ denote the path with $r+1$ vertices $1,\dots,r+1$ (so $P_r$ has $r$ edges). 
	\begin{lemma}\label{lem:path_counting_general_graphs}
		There is an algorithm which, given integers $r,\Delta \geq 1$ and an input graph $G$ with $m$ edges, computes the following:
		\begin{enumerate}
			\item For every $x,y \in V(G)$, the number
			$N_{r,\Delta}(x,y)$ of homomorphisms $\varphi: P_r \rightarrow G$ such that $\varphi(1) = x$, $\varphi(r+1) = y$, and $d_G(\varphi(t)) \leq \Delta$ for every $2 \leq t \leq r$. 
			\item For every $x,y \in V(G)$ with $d_G(x) > \Delta$ and for every function $f : \{2,\dots,r+1\} \rightarrow \{\textup{low, high}\}$, the number $M_{r,\Delta,f}(x,y)$ of homomorphisms $\varphi: P_r \rightarrow G$ such that $\varphi(1) = x$ and $\varphi(r+1) = y$, and such that for every $t \in \{2,\dots,r+1\}$ it holds that $d_G(\varphi(t)) \leq \Delta$ if $f(t) = \textup{low}$ and 
			$d_G(\varphi(t)) > \Delta$ if $f(t) = \textup{high}$.
		\end{enumerate} 
		The computation for Item 1 takes time $\tilde{O}(m\Delta^{r-1})$ and the computation for Item 2 takes time $\tilde{O}(m^2/\Delta)$. 
		Furthermore, the number of pairs $x,y$ with $N_{r,\Delta}(x,y) > 0$ is $O(m\Delta^{r-1})$. 
	\end{lemma}
	\begin{proof}
		As in Section \ref{sec:comb}, we save all computed information in a hash map. 
		We start by proving Item 1. Here we simply enumerate all homomorphisms 
		$\varphi: P_r \rightarrow G$ such that $d_G(\varphi(t)) \leq \nolinebreak \Delta$ for every $2 \leq t \leq r$. Enumerating all such homomorphisms is clearly sufficient to produce the desired counts. To choose such a homomorphism $\varphi$, we first choose the pair $(\varphi(1),\varphi(2))$, for which there are at most $2m$ choices, as this pair must be an edge. We only go over choices for which $d_G(\varphi(2)) \leq \Delta$. Having chosen $\varphi(2)$, we have at most $\Delta$ choices for $\varphi(3)$, where again we only consider choices which satisfy $d_G(\varphi(3)) \leq \Delta$. Having chosen $\varphi(3)$, there are at most $\Delta$ choices for $\varphi(4)$, and so on. Continuing in this fashion, we see that there are at most $\Delta$ choices for each of the vertices $\varphi(3),\dots,\varphi(r+1)$. Hence, the total number of choices is $O(m\Delta^{r-1})$. In particular, the number of pairs $x,y$ with $N_{r,\Delta}(x,y) > 0$ (namely, which are the endpoints in such a homomorphism) is $O(m\Delta^{r-1})$.
		
		We now establish Item 2 by induction on $r$. 
		The base case $r = 1$ is trivial, since $P_1$ is simply an edge, and we can enumerate all edges in time $m \leq m^2/\Delta$. (We may assume that $\Delta \leq m$, because otherwise there are no vertices of degree larger than $\Delta$.)
		
		We now proceed to the induction step. Assume, by the induction hypothesis, that we have already computed the desired counts for $P_{r-1}$. To achieve this for $P_r$, go over all choices for a vertex $x \in V(G)$ such that $d_G(x) > \Delta$ and a pair of vertices $(u,y)$ such that $\{u,y\} \in E(G)$. The number of such choices is at most $O(m^2/\Delta)$, since there are $2m$ choices for $(u,y)$ and $O(m/\Delta)$ choices for $x$. For each function 
		$g : \{2,\dots,r\} \rightarrow \{\textup{low, high}\}$, we have already computed the number $M_{r-1,\Delta,g}(x,u)$ of homomorphisms $\psi : P_{r-1} \rightarrow G$ such that $\psi(1) = x$ and $\psi(r) = u$, and such that for every $t \in \{2,\dots,r\}$ it holds that $d_G(\varphi(t)) \leq \Delta$ if $g(t) = \textup{low}$ and 
		$d_G(\varphi(t)) > \Delta$ if $g(t) = \textup{high}$.
		Now, for each such $\psi$ and $g$, the map 
		$\varphi : V(P_r) \rightarrow V(G)$ which agrees with $\psi$ on $P_r[\{1,\dots,r\}] = P_{r-1}$ and satisfies $\varphi(r+1) = y$, is a homomorphism from $P_r$ to $G$ satisfying $\varphi(1) = x$ and $\varphi(r+1) = y$. Furthermore, the function 
		$f : \{2,\dots,r+1\} \rightarrow \{\textup{low, high}\}$ corresponding to $\varphi$ is simply the function which agrees with $g$ on $\{2,3,\dots,r\}$ and satisfies $f(r+1) = \textup{low}$ if $d_G(y) \leq \Delta$ and $f(r+1) = \textup{high}$ if $d_G(y) > \Delta$. 
		It is now easy to see that in this manner we obtain the desired count for the pair $x,y$. 
	\end{proof}
	\begin{proof}[Proof of Theorem \ref{thm:hom_cnt_general_graphs}]
		We may assume that $G$ is connected, and hence $m \geq n - 1$, where $n = |V(G)|$. 
		Set $\Delta := m^{1/\lceil k/2 \rceil}$. 
		Denote the vertices of $C_k$ by $v_1,\dots,v_k$. In order to compute $\hom(C_k,G)$, we will compute, for every function 
		$g : \{v_1,\dots,v_k\} \rightarrow \{\textup{low, high}\}$, the number $\hom_g$ of homomorphisms 
		$\varphi : C_k \rightarrow G$ such that $d_G(\varphi(v_i)) \leq \Delta$ if $g(v_i) = \textup{low}$ and 
		$d_G(\varphi(v_i)) > \Delta$ if $g(v_i) = \textup{high}$ (for every $1 \leq i \leq k$). Clearly, $\hom(C_k,G) = \sum_{g}{\hom_g}$. We consider two cases. Suppose first that $g(v_i) = \textup{low}$ for every $1 \leq i \leq k$. It is easy to see that $C_k$ consists of a copy of $P_{\lfloor k/2 \rfloor}$ and a copy of $P_{\lceil k/2 \rceil}$, glued together at their endpoints. It is now easy to see that for this particular $g$, one has 
		$$
		\hom_g = \sum_{\substack{x,y \in V(\vec{G}): \\ d_G(x), d_G(y) \leq \Delta}}
		{N_{\lfloor k/2 \rfloor,\Delta}(x,y) \cdot 
			N_{\lceil k/2 \rceil,\Delta}(x,y)},
		$$
		where $N_{r,\Delta}(x,y)$ is as defined in Item 1 of Lemma \ref{lem:path_counting_general_graphs}. Since we can compute the counts 
		\linebreak $N_{\lfloor k/2 \rfloor,\Delta}(x,y), 
		N_{\lceil k/2 \rceil,\Delta}(x,y)$ for all pairs of vertices $x,y \in V(G)$ (simultaneously) in time $\tilde{O}(m \Delta^{\lceil k/2 \rceil-1})$, and since the total number of pairs for which these counts are non-zero is 
		$O(m \Delta^{\lceil k/2 \rceil-1})$, the above sum can be computed in time 
		$\tilde{O}(m \Delta^{\lceil k/2 \rceil-1}) = 
		\tilde{O}(m^{2 - 1/\lceil k/2 \rceil})$, as required.
		
		Suppose now that $g(v_i) = \textup{high}$ for some $1 \leq i \leq k$.
		We proceed similarly to the previous case. Decompose $C_k$ into a copy $P$ of $P_{\lfloor k/2 \rfloor}$ and a copy $P'$ of $P_{\lceil k/2 \rceil}$, such that $v_i$ is an endpoint of both of these paths. Let $f,f'$ be the ``low/high signatures" corresponding the paths $P,P'$, respectively. Again, it is not hard to see that 
		$$
		\hom_g = 
		\sum_{\substack{x,y \in V(G) \; : \; d_G(x) > \Delta}}
		{M_{\lfloor k/2 \rfloor,\Delta,f}(x,y) \cdot 
			M_{\lceil k/2 \rceil,\Delta,f'}(x,y) }.
		$$
		By Item 2 of Lemma \ref{lem:path_counting}, one can compute $M_{\lfloor k/2 \rfloor,\Delta,f}(x,y)$ and
		$M_{\lceil k/2 \rceil,\Delta,f'}(x,y)$ for all pairs 
		$x,y \in V(G)$ with $d_G(x) > \Delta$ in time $\tilde{O}(m^2/\Delta)$. Furthermore, the number of such pairs $x,y$ is $O(m^2/\Delta)$, because there are $O(m/\Delta)$ choices for $x$ and $n = O(m)$ choices for $y$. We conclude that the above sum can be computed in time 
		$\tilde{O}(m^2/\Delta) = \tilde{O}(m^{2 - 1/\lceil k/2 \rceil})$, as required. This completes the proof of the theorem. 
	\end{proof}

	\appendix
	
	\section{Proof of Lemma \ref{lem:hom-linear-combination}}\label{sec:matrix-Lovasz}
	
	
	\noindent
	The proof of Lemma \ref{lem:hom-linear-combination} uses the notion of {\em tensor product} of digraphs, which we now recall. 
	
	\begin{definition}[tensor product of directed graphs]\label{def:tensor-directed}
		Let $G_1$ and $G_2$ be directed graphs. The {\em tensor product} of $G_1$ and $G_2$, denoted $G_1 \times G_2$, is a directed graph with vertex set
		$V(G_1 \times G_2) = V(G_1) \times V(G_2)$, and edge set 
		
		$$E(G_1 \times G_2) = 
		\{((x_1,x_2), (y_1,y_2)) \, | \, (x_1,y_1) \in E(G_1) \textup{ and } (x_2,y_2) \in E(G_2)\}.$$
	\end{definition}
	
	\noindent A key property of the directed tensor product is that the parameter $\hom(\vec{H},\cdot)$ is multiplicative with respect to it (for any digraph $\vec{H}$). That is, for every pair of directed graphs $\vec{G_1}, \vec{G_2}$, it holds that
	
	\begin{equation}\label{eq:hom-multiplicative-for-tensor}
	\hom(\vec{H}, \vec{G_1} \times \vec{G_2}) = 
	\hom(\vec{H}, \vec{G_1}) \cdot \hom(\vec{H}, \vec{G_2}). 
	\end{equation}
	
	\noindent To see that \eqref{eq:hom-multiplicative-for-tensor} holds, simply observe that for functions $\varphi_i : V(\vec{H}) \rightarrow V(\vec{G_i})$ (where $i \in \{1,2\}$), the function $v \mapsto (\varphi_1(v),\varphi_2(v))$ is a homomorphism from $\vec{H}$ to $\vec{G_1} \times \vec{G_2}$ if and only if $\varphi_i$ is a homomorphism from $\vec{H}$ to $\vec{G_i}$ for each $i \in \{1,2\}$. 
	
	We will also need the following (trivial) observation regarding tensor products and degeneracy.
	
	\begin{observation}\label{obs:tensor-product-degrees}
		Let $\vec{F},\vec{G}$ be directed graphs. If the out-degree of every vertex in $G$ is at most $d$, then the out-degree of every vertex in 
		$\vec{F} \times \vec{G}$ is at most $v(\vec{F}) \cdot d$.
	\end{observation}
	
	\begin{proof}
		For each $x \in V(\vec{F})$ and $y \in V(\vec{G})$, the out-degree of $(x,y)$ in $\vec{F} \times \vec{G}$ is 
		$\text{out-deg}_{\vec{F} \times \vec{G}}((x,y)) = 
		\text{out-deg}_{\vec{F}}(x) \cdot \text{out-deg}_{\vec{G}}(y) < v(\vec{F}) \cdot d$.
	\end{proof}
	
	The last ingredient in the proof of Lemma \ref{lem:hom-linear-combination} is the following directed version of a lemma of Erd\H{o}s, Lov\'{a}sz and Spencer \cite{ELS} (see also Proposition 5.44(b) in \cite{Lovasz}).
	For completeness, we give its proof at the end of this appendix (this proof is essentially identical to the undirected case). 
	
	\begin{lemma}[\cite{Lovasz}]\label{lem:hom-matrix-invertible}
		Let $\vec{H_1}, \dots, \vec{H_k}$ be pairwise non-isomorphic directed graphs, and let 
		$c_1, \dots, c_k$ be non-zero constants. Then, there exist directed graphs $\vec{F_1}, \dots, \vec{F_k}$ such that the $k \times k$ matrix \linebreak
		$M_{i,j} = c_j \cdot \hom(\vec{H_j}, \vec{F_i})$, $1 \leq i,j \leq k$, is invertible. 
	\end{lemma}
	
	\begin{proof}[Proof of Lemma \ref{lem:hom-linear-combination}]
		The ``if" part of the lemma is immediate. Let us prove the ``only if" part. 
		So assume that 
		$\textsc{hom-cnt}_{c_1 \vec{H_1} + \dots + c_k \vec{H_k}}$ can be solved in time $O(n^\alpha)$ in $n$-vertex degenerate digraphs. Our goal is to show that $\textsc{hom-cnt}_{\vec{H_i}}$ can be solved in time $O(n^\alpha)$ for all $1 \leq i \leq k$. 
		Let $\vec{G}$ be an $n$-vertex degenerate digraph for which we would like to compute $\hom(\vec{H_i},\vec{G})$ for all $1 \leq i \leq \nolinebreak k$. By Lemma \ref{lem:hom-matrix-invertible}, there are digraphs $\vec{F_1}, \dots, \vec{F_k}$ such that the $k \times k$ matrix $M_{i,j} := c_j \cdot \hom(\vec{H_j},\vec{F_i})$ (for $1 \leq i,j \leq k$) is invertible. For each $1 \leq i \leq k$, we set 
		$b_i := c_1 \cdot \hom(\vec{H_1}, \vec{F_i} \times \vec{G}) + \dots + c_k \cdot \hom(\vec{H_k}, \vec{F_i} \times \vec{G})$ and observe that
		\begin{equation}\label{eq:tensor-product-bi}
		b_i = \sum_{j = 1}^{k}{c_j \cdot \hom(\vec{H_j},\vec{F_i} \times \vec{G})} = 
		\sum_{j = 1}^{k}{c_j \cdot \hom(\vec{H_j},\vec{F_i}) \cdot \hom(\vec{H_j},\vec{G})} = 
		\sum_{j = 1}^{k}{M_{i,j} \cdot \hom(\vec{H_j},\vec{G})}
		\end{equation}
		where we have used \eqref{eq:hom-multiplicative-for-tensor}. 
		We will treat \eqref{eq:tensor-product-bi} (for $1 \leq i \leq k$) as a system of linear equations, where $\hom(\vec{H_1},\vec{G}),\dots,\hom(\vec{H_k},\vec{G})$ are the variables, $M$ is the matrix of the system, and $b_1,\dots,b_k$ are the constant terms. Since $M$ is invertible (as guaranteed by our choice of $\vec{F_1},\dots,\vec{F_k}$), knowing $b_1,\dots,b_k$ would enable us to find $\hom(\vec{H_1},\vec{G}),\dots,\hom(\vec{H_k},\vec{G})$.
		
		Note that for each $1 \leq i \leq k$, the digraph $\vec{F_i} \times \vec{G}$ can be constructed in time $O(n)$, is $O(1)$-degenerate (see Observation \ref{obs:tensor-product-degrees}), and has $v(\vec{F_i}) \cdot n = O(n)$ vertices. Hence, since by our assumption 
		$\textsc{hom-cnt}_{c_1 \vec{H_1} + \dots + c_k \vec{H_k}}$ can be solved in time $O(n^\alpha)$, we can compute $b_1, \dots, b_k$ in time $O(n^\alpha)$ by feeding the graphs $\vec{F_1} \times \vec{G}, \dots, \vec{F_k} \times \vec{G}$ to an algorithm which solves $\textsc{hom-cnt}_{c_1 H_1 + \dots + c_k H_k}$ in time $O(n^\alpha)$ in degenerate digraphs. This in turn enables us to compute $\hom(\vec{H_1},\vec{G}), \dots, \hom(\vec{H_k},\vec{G})$, as required. 	
	\end{proof} 
	
	\begin{proof}[Proof of Lemma \ref{lem:hom-matrix-invertible}]
		Since multiplying the rows of an invertible matrix by non-zero scalars leaves it invertible, we may assume, without loss of generality, that $c_1 = \dots = c_k = 1$. 
		We will also assume \nolinebreak that
		\begin{equation}\label{eq:order-of-H_i}
		v(H_1)+e(H_1) \leq v(H_2)+e(H_2) \leq \dots \leq v(H_k)+e(H_k),
		\end{equation} 
		where, as always, $v(G), e(G)$ denote the number of vertices and the number of edges in $G$, respectively. 
		For every $1 \leq i \leq k$, we define a suitable \textit{blowup} of $H_i$, denoted $F_i$, as follows. Every vertex $v \in V(H_i)$ is replaced by an independent set $I_{i,v}$ of size $x_{i,v} \geq 1$, and every edge $(u,v) \in E(H_i)$ is replaced by a directed complete bipartite graph, connecting $I_{i,u}$ to $I_{i,v}$. The resulting graph is $F_i$. We will consider the values  
		$\{x_{i,v} \, | \, i \in [k], v \in V(H_i)\}$ 
		as variables taking positive integer values, and show that there is an assignment to these variables for which the matrix  
		$M_{i,j} = \textup{hom}(H_i, F_j)$ is invertible. (This matrix is the transpose of the matrix appearing in the statement of the lemma, so this will be sufficient.)
		
		We claim that $D := \det(M)$ is a non-zero polynomial in the variables 
		$\{x_{i,v} \, | \, i \in [k], v \in V(H_i)\}$. In what follows, it will be convenient to use the following notation: for graphs $H,G$, $\text{Hom}(H,G)$ denotes the set of all homomorphisms from $H$ to $G$, and $\text{Inj}(H,G)$ denotes the set of all injective homomorphisms from $H$ to $G$ (so $\hom(H,G) = |\text{Hom}(H,G)|$ and $\inj(H,G) = |\text{Inj}(H,G)|$). 
		We start by observing that for all $i,j \in [k]$,
		\begin{equation}\label{eq:hom(H_i,F_j)}
		\hom(H_i, F_j) = \sum_{\varphi \in \text{Hom}(H_i,H_j)} \,
		\prod_{v \in V(H_j)} x_{j,v}^{|\varphi^{-1}(v)|} \,.
		\end{equation} 
		Indeed, it is easy to see that every homomorphism $\psi : H_i \rightarrow F_j$ corresponds to some homomorphism $\varphi: H_i \rightarrow H_j$ (if $\alpha_j: F_j \rightarrow H_j$ denotes the ``natural map" mapping $I_{j,v}$ to $v$ for all $v \in V(H_j)$, then $\varphi$ can be expressed as $\varphi = \alpha_j \circ \psi$). Moreover, the number of homomorphisms $\psi: H_i \rightarrow F_j$ which correspond to a given homomorphism $\varphi: H_i \rightarrow H_j$ is exactly the number of ways to choose a vertex from $I_{j,\varphi(u)}$ for each $u \in V(H_i)$; this number is exactly the product on the right-hand side of \eqref{eq:hom(H_i,F_j)}.
		
		Since every entry of $M$ is a polynomial in $\{x_{i,v} \, | \, i \in [k], v \in V(H_i)\}$, so is $D = \det(M)$. 
		Our goal is to prove that $D$ is not the zero polynomial. To this end, we will show that the multilinear part of $D$ is non-zero (which clearly implies that $D$ is non-zero). For convenience, let us write $ML(p)$ for the multilinear part of a polynomial $p$ (i.e., the sum of all multilinear monomials of $p$).
		Observe that for every pair $i,j \in [k]$ and for every homomorphism $\varphi: V(H_i) \rightarrow V(H_j)$, 
		the monomial 
		$$
		\prod_{v \in V(H_j)} x_{j,v}^{|\varphi^{-1}(v)|}
		$$ 
		is multilinear if and only if $|\varphi^{-1}(v)| \leq 1$ for all $v$, which is equivalent to $\varphi$ being injective. Thus, using \eqref{eq:hom(H_i,F_j)}, we conclude that the multilinear part of $\hom(H_i,F_j)$ is 
		\begin{equation}\label{eq:multilinear-L_i_j}
		L_{i,j} := ML(\hom(H_i,F_j)) = \sum_{\varphi \in \text{Inj}(H_i,H_j)} \, \prod_{v \in V(H_j)} x_{j,v}^{|\varphi^{-1}(v)|} \,.
		\end{equation} 
		Observe that for each $1 \leq i \leq k$, $L_{i,i} \neq 0$ (as a polynomial). Indeed, injective homomorphisms from $H_i$ to itself are simply automorphisms of $H_i$, so \eqref{eq:multilinear-L_i_j} becomes
		$$
		L_{i,i} = \aut(H_i) \cdot \prod_{v \in V(H_i)}{x_{i,v}}.
		$$
		We now claim that for every pair $1 \leq j < i \leq k$, one has $\text{Inj}(H_i,H_j) = \emptyset$, and hence $L_{i,j} \equiv 0$. Fix any $1 \leq j < i \leq k$. If there exists an injective homomorphism from $H_i$ to $H_j$, then we must have $v(H_i) \leq v(H_j)$ and $e(H_i) \leq e(H_j)$. On the other hand, we have assumed in (\ref{eq:order-of-H_i}) that $v(H_j)+e(H_j) \leq v(H_i)+e(H_i)$, as $j < i$. It follows that $v(H_i) = v(H_j)$ and $e(H_i) = e(H_j)$, which implies that $H_i$ and $H_j$ are isomorphic, contradicting our assumption that $H_1, \dots, H_k$ are pairwise non-isomorphic. Therefore, $\text{Inj}(H_i,H_j) = \emptyset$, which immediately implies that $L_{i,j} \equiv 0$. 
		
		Now consider the $k \times k$ matrix $L$ whose entries are $L_{i,j}$ ($1 \leq i,j \leq k$). In the previous paragraph we established that $L$ is 
		upper-triangular, and that its diagonal entries are non-zero (as polynomials). It follows that $\det(L)$ is a non-zero polynomial. Finally, we show that the multilinear part of \linebreak $D = \det(M)$ is $ML(D) = \det(L)$.
		To this end, we recall that
		\begin{equation}\label{eq:det(M)_1}
		D = \sum_{\sigma \in S_k} \text{sgn}(\sigma) \cdot \prod_{i=1}^k M_{i,\sigma(i)} = 
		\sum_{\sigma \in S_k} \text{sgn}(\sigma) \cdot \prod_{i=1}^k \hom(H_{i}, F_{\sigma(i)}).
		\end{equation} 
		Now, since the sets of variables of $\hom(H_1,F_{\sigma(1)}), \dots, \hom(H_k,F_{\sigma(k)})$ are pairwise-disjoint (recall that the variables of $\hom(H_i,F_j)$ are $x_{j,v}$, $v \in V(H_j)$), we have
		\begin{equation}\label{eq:det(M)_2}
		ML\left( \prod_{i=1}^k \hom(H_{i}, F_{\sigma(i)}) \right) = 
		\prod_{i=1}^k ML(\hom(H_{i}, F_{\sigma(i)})) = 
		\prod_{i=1}^k L_{i,\sigma(i)}.
		\end{equation}
		By combining \eqref{eq:det(M)_1} and \eqref{eq:det(M)_2}, we get
		$$
		ML(D) = \sum_{\sigma \in S_k} \text{sgn}(\sigma) \cdot \prod_{i=1}^k L_{i,\sigma(i)} = \det(L),
		$$
		as required. As $\det(L)$ is a non-zero polynomial, $D$ is a non-zero polynomial as well.  
		It follows that there is an assignment of positive natural numbers to the variables $\{x_{i,v} \, | \, i \in [k], v \in V(H_i)\}$ for which $D \neq 0$. Evidently, for this assignment, the matrix $M$ is invertible (as $D = \det (M)$). This completes \nolinebreak the \nolinebreak proof.  
	\end{proof}
	


\begin{thebibliography}{99}
		\bibitem{A-VW}
		A. Abboud and V. Vassilevska Williams, Popular conjectures imply strong lower
		bounds for dynamic problems. In Proc. 55th Annual IEEE Symposium on Foundations of Computer Science, 2014.
		
		\bibitem{AV}
		J. Alman and V. Vassilevska Williams, A refined laser method and faster matrix multiplication. In Proceedings of the 2021 ACM-SIAM Symposium on Discrete Algorithms (SODA) (pp. 522-539). Society for Industrial and Applied Mathematics.
		
		\bibitem{AYZ_ColorCoding}
		N. Alon, R. Yuster and U. Zwick, Color-coding. Journal of the ACM (JACM), 42(4):844-856, 1995. 
		
		\bibitem{AYZ}
		N. Alon, R. Yuster and U. Zwick, Finding and counting given length cycles. Algorithmica, 17(3):209-223, 1997.
		
		\bibitem{BA}
		A. L. Barab\'{a}si and R. Albert, Emergence of scaling in random networks. Science, 286(5439), 509-512, 1999. 
		
		\bibitem{Bera-Sesh-GLS}
		S. K. Bera, L. Gishboliner, Y. Levanzov, C. Seshadhri and A. Shapira, Counting subgraphs in degenerate graphs. ACM Journal of the ACM (JACM) 69, no. 3, pp. 1-21, 2022.
		
		\bibitem{BPS}
		S. K. Bera, N. Pashanasangi and C. Seshadhri, Linear Time Subgraph Counting, Graph Degeneracy, and the Chasm at Size Six. In Proceedings of the 11th Innovations in Theoretical Computer Science Conference (ITCS 2020), 38:1-38:20.
		
		\bibitem{BPS_2}
		S. K. Bera, N. Pashanasangi and C. Seshadhri, Near-linear time homomorphism counting in bounded degeneracy graphs: the barrier of long induced cycles. In Proceedings of the 2021 ACM-SIAM Symposium on Discrete Algorithms (SODA) (pp. 2315-2332). Society for Industrial and Applied Mathematics.
		
		\bibitem{Bollobas}
		B. Bollob\'{a}s, On generalized graphs. Acta Mathematica Academiae Scientarium Hungaricae 16, 447-452, 1965.
		
		\bibitem{Bressan} M. Bressan, Faster algorithms for counting subgraphs in sparse graphs. Algorithmica, pp. 1-28, 2021.
		
		\bibitem{CN}
		N. Chiba and T. Nishizeki. Arboricity and subgraph listing algorithms. SIAM Journal on computing, 14(1), pp. 210-223, 1985.
		
		\bibitem{CDM}
		R. Curticapean, H. Dell and D. Marx, Homomorphisms are a good basis for
		counting small subgraphs. In Proceedings of the 49th Annual ACM SIGACT Symposium on Theory of Computing, pages 210-223, 2017.
		
		\bibitem{C_Marx}
		R. Curticapean and D. Marx, Complexity of counting subgraphs: Only the boundedness of the vertex-cover number counts. In Proc. 55th Annual IEEE Symposium on Foundations of Computer Science, 130-139, 2014.
		
		\bibitem{DVW}
		M. Dalirrooyfard, T. D. Vuong and V. Vassilevska Williams, Graph pattern detection: Hardness for all induced patterns and faster non-induced cycles. In Proceedings of the 51st Annual ACM SIGACT Symposium on Theory of Computing, pp. 1167--1178, 2019.
		
		\bibitem{DJ}
		V. Dalmau and P. Jonsson, The complexity of counting homomorphisms seen from the other side. Theoretical Computer Science, 329(1-3), 315-323, 2004.‏
		
		\bibitem{EG}
		F. Eisenbrand and F. Grandoni,
		Detecting directed 4-cycles still faster. Information Processing Letters, 87(1), pp.13-15, 2003. 
		
		\bibitem{ELS}
		P. Erd\H{o}s, L. Lov\'{a}sz and J. Spencer, Strong independence of graphcopy functions. Graph Theory and Related Topics, 165--172, 1979. 
		
		\bibitem{FG}
		J. Flum and M. Grohe, The parameterized complexity of counting problems. In Proc. 43rd IEEE Symposium on Foundations of Computer Science, 2002, pp. 538–547.
		
		\bibitem{Grohe}
		M. Grohe, The complexity of homomorphism and constraint satisfaction problems seen from the other side. Journal of the ACM (JACM), 54(1):1–24, 2007.
		
		\bibitem{HP}
		X. Huang and V. Y. Pan, Fast rectangular matrix multiplication and applications. Journal of complexity, 14(2), pp. 257--299, 1998. 
		
		\bibitem{IR}
		A. Itai and M. Rodeh, Finding a minimum circuit in a graph. SIAM Journal on Computing, 7(4):413-423, 1978.
		
		\bibitem{L-VW-W}
		A. Lincoln, V. Vassilevska Williams and R. Williams, 
		Tight hardness for shortest cycles and paths in sparse graphs. In Proceedings of the Twenty-Ninth Annual ACM-SIAM Symposium on Discrete Algorithms, 1236-1252, 2018.
		
		\bibitem{Legall}
		F. Le Gall, Powers of tensors and fast matrix multiplication. In International Symposium on
		Symbolic and Algebraic Computation (ISSAC 14'), 296-303, 2014.
		
		\bibitem{Lovasz}
		L. Lov\'{a}sz, {\bf Large networks and graph limits}, 2012, Providence: American Mathematical Society.
		
		\bibitem{Marx}
		D. Marx, Can you beat treewidth? In 48th Annual IEEE Symposium on Foundations of Computer Science (FOCS '07), 169-179, 2007.
		
		\bibitem{MB}
		D. W. Matula, and L. L. Beck, Smallest-last ordering and clustering and graph coloring algorithms. Journal of the ACM (JACM), 30(3), pp. 417-427, 1983.
		
		\bibitem{Meeks}
		K. Meeks, The challenges of unbounded treewidth in parameterised subgraph counting problems. Discrete Applied Mathematics, 198, 170-194, 2016.
		
		\bibitem{MSIKCA}
		R. Milo, S. Shen-Orr, S. Itzkovitz, N. Kashtan, D. Chklovskii and U. Alon,
		Network motifs: simple building blocks of complex networks. Science, 298(5594), 824-827, 2002.
		
		\bibitem{Monien}
		B. Monien, How to find long paths efficiently. Annals of Discrete Mathematics, 25:239-254, 1985. 
		
		\bibitem{NDeM}
		J. Ne\v{s}et\v{r}il and P. O. De Mendez, {\bf Sparsity: graphs, structures, and algorithms} (Vol. 28), Springer Science \& Business Media, 2012. 
		
		\bibitem{NP}
		J. Ne\v{s}et\v{r}il and S. Poljak, On the complexity of the subgraph problem. Commentationes Mathematicae Universitatis Carolinae, 26(2):415-419, 1985.
		
		\bibitem{Przulj}
		N. Przulj, Biological network comparison using graphlet degree distribution. Bioinformatics, 23(2):177-183, 2007.
		
		\bibitem{RW}
		S. Rudich and A. Wigderson, {\bf Computational complexity theory}, (Vol. 10), American Mathematical Soc., 2004. 
		
		\bibitem{V W}
		V. Vassilevska Williams, Efficient algorithms for clique problems. Information Processing Letters, 109(4):254-257, 2009.
		
		\bibitem{V W 2}
		V. Vassilevska Williams, Multiplying matrices faster than Coppersmith-Winograd. In Proceedings of the 44th Symposium on Theory of Computing Conference (STOC), 887-898, 2012.
		
		\bibitem{VW-W}
		V. Vassilevska Williams and R. Williams, Finding, minimizing, and counting weighted subgraphs. In Proc. 41st Annual ACM Symposium on the Theory of Computing, 455-464, 2009.
		
		\bibitem{VW-W-Y}
		V. Vassilevska Williams, R. Williams and R. Yuster, Finding heaviest $H$-subgraphs in real weighted graphs, with applications. ACM Transactions on Algorithms (TALG), 6(3), pp.1-23, 2010.
		
		\bibitem{YZ}
		R. Yuster and U. Zwick, Detecting short directed cycles using rectangular matrix multiplication and dynamic programming. In SODA, vol. 4, pp. 254--260, 2004.
		
		\bibitem{YZ_2}
		R. Yuster and U. Zwick, Finding even cycles even faster. SIAM Journal on Discrete Mathematics, 10(2), 209-222, 1997.
	\end{thebibliography}
\end{document}